\documentclass[12pt]{article}
\usepackage[english]{babel}
\usepackage[utf8]{inputenc}
\usepackage{amsmath,amsfonts,amssymb,amsthm}
\usepackage{enumerate}
\usepackage{setspace}
\usepackage[plain]{fullpage}
\usepackage[authoryear, longnamesfirst]{natbib}
\usepackage{mathrsfs}  
\usepackage{calligra}
\usepackage{hyperref}
\usepackage{adjustbox}
\usepackage{subcaption}
\usepackage{color}
\usepackage{graphicx}
\usepackage{algorithm}
\usepackage{algpseudocode}
\usepackage{footnotehyper}
\makesavenoteenv{algorithm}
\usepackage{amsmath,amssymb} \usepackage{comment}

\usepackage{booktabs}

\theoremstyle{definition}
\newtheorem{assumption}{Assumption}
\newtheorem{definition}{Definition}

\theoremstyle{plain}

\newtheorem{thm}{Theorem}
\newtheorem{lemma}{Lemma}

\theoremstyle{remark}
\newtheorem{remark}{Remark}

\title{\textbf{Epsilon-Minimax and Epsilon-Maximin \\Solutions of Statistical Decision Problems}\footnote{We would like to thank Karun Adusumilli, Isaiah Andrews, Tim Armstrong, Kevin Chen, Paul Delatte, Giannis Fikioris, Patrik Guggenberger, Kei Hirano, Nicole Immorlica, Michal Koles\'ar, Yoav Kolumbus, Lihua Lei, Charles Manski, Francesca Molinari, Ulrich M\"uller, Guillaume Pouliot, Brenda Quesada Prallon, Anant Shah, David Shmoys, Karthik Sridharan, Sophie Sun, Yiwei Sun, Keita Sunada, Vasilis Syrgkanis, Eva Tardos, Alex Torgovitsky, Davide Viviano, audiences at Princeton, the JER/SNSF 2025 workshop in econometrics, the Southern Economic Association 2025 Annual Meeting, and the Canadian Econometrics Study Group 40$^{th}$ Meeting, as well as four anonymous referees and the audience at the Twenty-Sixth ACM Conference on Economics and Computation (EC'25) for helpful feedback, comments, and suggestions. We would also like to thank Gil Peled, Lawrence Ozoigbo, and Rohit Kumar for excellent research assistance. We gratefully acknowledge financial support from the NSF under grants SES-2315600, 2229012, 2312204, and 2403007 and from the Department of Defense through the Air Force Office of Scientific Research under award numbers FA9550-20-1-0397 and ONR 1398311.}}

\author{Andr\'es Aradillas Fern\'andez\thanks{Department of Economics, Massachusetts Institute of Technology.} \and Jos\'e Blanchet\thanks{Management Science and Engineering Department, Stanford University.} \and Jos\'e Luis Montiel Olea\thanks{Department of Economics, Cornell University.} \and Chen Qiu\footnotemark[4] \and J\"org Stoye\footnotemark[4] \and Lezhi Tan\footnotemark[3] 
}
\date{}

\begin{document}
\maketitle

\onehalfspacing

\begin{abstract}
A decision rule is \emph{$\epsilon$-minimax} if it is minimax up to an additive factor $\epsilon$. A prior distribution is \emph{$\epsilon$-maximin} if it is \emph{least favorable} up to an additive factor $\epsilon$. We present an algorithm to provably obtain $\epsilon$-minimax and $\epsilon$-maximin solutions for a class of statistical decision problems. In particular, we are interested in problems in which either the statistician or nature can randomly choose among $I$ strategies. The minimax solution of these problems admits a convex programming representation over the $(I-1)$-simplex. Our suggested algorithm is a well-known  \emph{mirror subgradient descent (or ascent)} routine, designed to approximately solve the convex optimization problem that defines the minimax/maximin problem. This iterative routine is known in the computer science literature as the \emph{Hedge algorithm} and is used in algorithmic game theory as a practical tool to find approximate solutions of two-person zero-sum games. We apply the suggested algorithm to different minimax problems in the econometrics literature. An empirical application to the problem of optimally selecting sites to maximize the external validity of an experimental policy evaluation illustrates the usefulness of the suggested procedure. 
\end{abstract}

\doublespacing

\section{Introduction}

Under \cite{Wald50}'s \emph{minimax} criterion, different statistical decision rules are ranked based on their worst possible expected loss. Searching for a \emph{minimax-optimal} decision rule---i.e., a rule with the smallest worst-case expected loss---is computationally challenging. It is known that obtaining the minimax solution of a decision problem---and sometimes even deciding whether a minimax solution exists---is NP-hard in general \citep{du1995minimax}. Moreover, the problem of searching for a minimax-optimal decision rule remains complicated even when the \emph{minimax theorem} holds. In this case, it is common to try to obtain a minimax decision rule indirectly by solving the \emph{maximin} problem, where an adversarial nature searches for a probability distribution over the parameter space of the statistical model---the \emph{least-favorable distribution}---to inflict the largest possible average risk on the decision maker. 

In this paper, we consider a particular class of decision problems in which either the decision maker or the adversarial nature is restricted to choosing from a menu of $I$ available strategies ($I$ decision rules for the decision maker or $I$ parameter values for nature), all of which are assumed to have risk between zero and a known positive constant $M$. When we allow the decision maker to choose randomly among $I$ options, the corresponding minimax problem can be viewed as a nonlinear convex minimization problem on the $(I-1)$-dimensional simplex. Similarly, when we allow an adversarial nature to choose randomly among the $I$ parameter values, the corresponding maximin problem can be viewed as a nonlinear concave maximization problem over the $(I-1)$-dimensional simplex. The connection between minimax/maximin problems and convex optimization is well known in the literature; see, for example,  \cite{chamberlain2000econometric}.  

We show that it is possible to make substantial progress in solving our class of statistical decision problems if, instead of insisting on finding an \emph{exact} minimax/maximin solution, we make it our goal to find an \emph{approximate} minimax/maximin solution. In particular, we search for either i) \emph{a rule that attains the smallest worst-case expected loss, but up to a given additive factor $\epsilon$} or ii) \emph{a distribution that is least favorable, but up to a given additive factor $\epsilon$}. The statistical decision theory literature refers to i) as an \emph{$\epsilon$-minimax optimal decision rule} \citep[Chapter 1, Definition 4, p.33]{Ferguson67}. We refer to the distribution that satisfies ii) as an \emph{$\epsilon$-maximin optimal distribution} or \emph{$\epsilon$-least-favorable distribution}. Throughout the rest of the paper we focus primarily on the problem of finding an $\epsilon$-minimax decision rule. But our results apply immediately to both problems (in the same way that results for convex minimization translate immediately to results regarding concave maximization).

Our main result shows that we can \emph{provably} obtain an $\epsilon$-minimax rule by using a \emph{mirror subgradient descent routine} for convex optimization (Theorem 1). The methods of mirror descent \citep[Chapter 3]{nemirovski1983problem} are a family of iterative procedures recommended in the optimization literature for approximately solving convex problems of high dimensions. These methods are iterative, first-order optimization algorithms, in that they require repeated evaluations of the objective function and its subgradient but do not exploit any further smoothness information about the objective function.  

We present an explicit upper bound on the number of evaluations of the objective function and its subgradient required by our suggested algorithm. In particular, we show that it suffices to stop the mirror subgradient descent routine after $T = \lceil2 M^2\ln(I) / \epsilon^2 \rceil$ iterations.\footnote{$\lceil \cdot \rceil$ is the ceiling function: the function that returns the smallest integer that is greater than or equal to a given number.} We use results in \cite{ben2001ordered} to argue that the smallest number of iterations required by \emph{any} iterative, first-order algorithm for finding an $\epsilon$-minimax rule is $O(1)M^2/ \epsilon^2$, provided $\epsilon \geq M / \sqrt{I}$. Thus, there is a sense in which the recommended algorithm, and the suggested number of iterations, achieve the optimal dependence on $M$ and $\epsilon$, up to the logarithmic factor $\ln(I)$. When applied to maximin problems, mirror ascent provides a natural, and theoretically grounded alternative, to the sequential quadratic programming algorithm used by \cite{chamberlain2000econometric} to find approximate least-favorable distributions in maximin problems.   

The algorithm herein suggested is known in the computer science literature as the \emph{Hedge algorithm} (a particular case of the \emph{Multiplicative Weights} update method); see Section 2.1 in \cite{arora2012multiplicative}. This method is used in problems where a decision maker chooses randomly among $I$ alternatives repeatedly (an \emph{online decision-making} problem), but after each round he obtains a payoff for all of the $I$ available actions. The Hedge algorithm is commonly used in algorithmic game theory as a practical tool to find approximate solutions of two-person zero-sum games. To the best of our knowledge, the explicit use of the Hedge algorithm in statistical decision problems is novel.  This is rather surprising in light of the straightforward connection between statistical decision problems and two-person zero-sum games, and the origins of Multiplicative Weights in iterative dynamics for game play---see the notion of $\kappa$-exponential fictitious play in \cite{fudenberg1995consistency} and the references to the work of \cite{blume1993statistical} therein.\footnote{  \cite{freund1999adaptive} use the Hedge algorithm to approximately solve the mixed extension of two-person zero-sum games where both players have finitely many pure strategies. However, for games in which one player has infinitely many pure strategies, some other algorithms have been suggested in the literature; see, for example, \cite{filar1982algorithm} and our discussion of related literature below.} 

Although our suggested algorithm is designed to \emph{directly} obtain $\epsilon$-minimax solutions of statistical decision problems, we show that the algorithm's output can also be used to provably construct---\emph{indirectly}---an \emph{$\epsilon$-maximin} solution; see Remark \ref{remark:LFD}. As we explain later, the indirect $\epsilon$-maximin solution will have at most $\lceil2 M^2\ln(I) / \epsilon^2 \rceil$ support points, even if the parameter space of nature has infinitely many points. We note that this \emph{indirect} use of the mirror descent algorithm is different from the \emph{direct} use of our algorithm to find an $\epsilon$-maximin solution in a problem where the parameter space has been discretized (in which case, we can simply do mirror ascent based on the maximin problem).

While we think there are cases (as in our applications) in which such a set of finitely many decision rules can be chosen naturally, there are precedents in the literature that view such finite set as arising from a discretization of the space of decision rules; see, for example,  \cite{ghosh1964uniform} and the references therein.\footnote{For instance, Section 5 in \cite{ghosh1964uniform} gives an illustration of the construction of one such finite set of decision rules when estimating a bounded mean of a normal distribution under quadratic loss.} While there are cases in which the error from the discretization can be derived explicitly (as we explain in Appendix \ref{app:decision-discretization}), we think whenever the problem of interest has infinitely many decision rules it is more natural to use the $\epsilon$-minimax/maximin solutions to provide upper/lower bounds on the minimax value of the problem of interest; see Remark \ref{remark:infiniteD}.  

We illustrate the usefulness of our suggested algorithm by analyzing a simple and stylized \emph{binary treatment choice problem with partial identification} based on the work of \cite{stoye2012minimax}. We use this example to compare the output of our algorithm with known exact solutions. We consider two types of minimax problems: minimizing worst-case regret and minimizing worst-case Bayes risk using the class of priors in \cite{giacomini2021robust}.

Finally, we present an empirical application to the problem of optimally selecting \emph{sites} to maximize the external validity of an experimental policy evaluation. This \emph{site selection problem} has been recently introduced in the work of \cite{gechter2023site} and \cite{egamidesigning}. Broadly speaking, a policy maker wishes to experimentally evaluate the effects of a new policy with the end goal of recommending its implementation on a set of different \emph{sites}. There are two types of sites: \emph{policy-relevant} and \emph{experimental} sites. There are also covariates $X_s \in \mathbb{R}^d$ available for each site. The \emph{site selection problem} asks the following question: if the policy maker can pick at most $k$ experimental sites, what are the sites that optimize external validity? Under assumptions we make precise, our approach provides an algorithm for deciding how to \emph{randomly} select sites to approximately optimize external validity, taking into account information about baseline covariates. When the policy maker is restricted to select only one site for experimentation, the output of our algorithm is a selection probability for each of the sites available for experimentation. Our application has two main messages. First, choosing uniformly at random where to experiment does not tend to be $\epsilon$-minimax optimal. Instead, the $\epsilon$-minimax solution adjusts the probability of sampling a site based on its baseline covariates. Second, there are cases---for example, when one experimental site is closest to each of the policy-relevant sites---in which the $\epsilon$-minimax solution assigns probability close to one to such most \emph{representative} site.

{\scshape Related Literature:} Different algorithms have been suggested for approximating the solutions of minimax problems like the ones considered in this paper. Some references include \cite{troutt1978algorithms}, \cite{ filar1982algorithm}, \cite{kempthorne1987numerical}, \cite{chamberlain2000econometric}, \cite{elliott2015nearly}, \cite{gupta2020learning}, and \cite{guggenberger2025numerical}. In general, we think there are two important differences between our work and this existing literature. First, instead of just focusing on testing problems or estimation problems, our analysis applies to a general class of decision problems. Second---once a desired approximation error $\epsilon$ has been selected, and once the bound $M$ on the risk function has been obtained---there are no further inputs that the user needs to specify in order to run the algorithm. This means that we are explicit about the number of iterations, step size, and also the initial condition. Importantly, we are able to guarantee that, upon termination after our suggested number of rounds, the algorithm provably generates an $\epsilon$-minimax rule---in the sense of \cite{Ferguson67}---provided our assumptions are satisfied. As discussed before, there is also a sense in which our algorithm is, up to a logarithmic term, as good as any other iterative, first-order algorithm.  

We next provide some more specific comments regarding the relation between our work and that of \cite{chamberlain2000econometric}, \cite{elliott2015nearly}, and \cite{guggenberger2025numerical}.

\noindent \emph{$\bullet$ Relation to \cite{chamberlain2000econometric}:} As mentioned earlier, the connection we establish between minimax/maximin problems and convex programming is due entirely to \cite{chamberlain2000econometric}.   \cite{chamberlain2000econometric}  focuses on maximin problems with a discretized parameter space and proposes a quadratic programming  algorithm  to find approximate least-favorable distributions. The ``Chamberlain algorithm'' in the literature is, in fact, typically equated with the practice of i) discretizing the parameter space, ii) obtaining a candidate least-favorable prior via concave programming, and iii) reporting the Bayes rule based on such  a prior. But the main difference between our work and that of \cite{chamberlain2000econometric} is not about whether we discretize the space of decision rules or the parameter space to obtain (either directly or indirectly) an approximate minimax solution. Instead, our main contribution is that we provide a theoretical analysis of the specific algorithm for convex/concave programming that is used after discretization (based on existing results in the literature), and use the specifics of the algorithm to attach theoretical guarantees to its output (either for an approximate minimax decision or a least favorable distribution). None of this analysis is covered in the work of \cite{chamberlain2000econometric}, \cite{elliott2015nearly} or \cite{guggenberger2025numerical}. 
  
  An important observation we also make in our work is that---without further qualification---a Bayes response to an approximate least favorable distribution is not guaranteed to be approximately minimax optimal. We illustrate this point using our treatment choice example with partially identified parameters (Section \ref{subsection:MMR}) where the best response
to a provable $\epsilon$-least-favorable prior fails to deliver an $\epsilon$-minimax rule. We hope to clarify that---without further assumptions---the pipeline of the Chamberlain algorithm need not work to generate
a provable approximate minimax rule.   
\noindent \emph{$\bullet$ Relation to \cite{elliott2015nearly}:} An important distinction vis-\'a-vis the influential work of \cite{elliott2015nearly} is that they focus on obtaining least-favorable distributions for hypothesis testing problems, whereas we focus on obtaining approximate minimax/maximin solutions for a general class of decision problems.
 The notion of least favorable distribution in  \cite{elliott2015nearly} arises from size and power considerations that are specific to testing problems, and differ from our notion of $\epsilon$-minimaxity. More details regarding this difference are provided in the work of \cite{fernandez2025approximate}, who show that the algorithm suggested by \cite{elliott2015nearly} for finding an approximately least-favorable distribution for hypothesis testing problems can be viewed as an application of stochastic mirror descent. But it is important to mention that there are no theoretical results in the work of \cite{elliott2015nearly} that could be used to derive any of our results.

\noindent \emph{$\bullet$ Relation to \cite{guggenberger2025numerical}:} More recently, \cite{guggenberger2025numerical} have shown that it is possible to obtain numerical approximations to minimax regret treatment rules in certain treatment choice problems by using a \emph{fictitious play} algorithm. In Appendix \ref{app:gp} we present a detailed comparison of the numerical performance of our algorithm and the fictitious play algorithm proposed by \cite{guggenberger2025numerical} in two concrete examples. While the two algorithms perform well in our examples, our algorithm can provably obtain $\epsilon$-maximin/minimax solutions in a wide range of decision problems. As noted by \cite{guggenberger2025numerical}, there are
examples where fictitious play can perform poorly; see \cite{daskalakis2014counterexample}.
  
{\scshape Outline:} The rest of the paper is organized as follows. Section 2 introduces notation, main assumptions, and presents the convex programming representation of the minimax problems analyzed herein. Section 3 defines an $\epsilon$-minimax decision rule and presents the algorithm. Section 4 applies the algorithm to two illustrative examples that involve solving treatment choice problems with partial identification. Our algorithm is then used to solve for $\epsilon$-minimax (regret) optimal rules; but we also argue that it can be applied to solve other minimax problems, such as (ex-ante) Robust Bayes analysis with the priors suggested by \cite{giacomini2021robust}. Section 5 presents the main application. Section 6 concludes.  
 
\section{Minimax Problems}

\subsection{Notation}
A decision maker must choose an action $a$ that belongs to some set $\mathcal{A}$. Prior to choosing the action, he observes data: the realization of a random variable $X$ taking values in a set $\mathcal{X}$. A data-driven choice of action is summarized by a \emph{decision rule}: a mapping from data to actions, which is herein denoted by the function $d:\mathcal{X} \rightarrow \mathcal{A}$.

We restrict our analysis to the case in which the decision maker only considers $I$ decision rules that belong to the finite set $\mathcal{D} \equiv \{d_1, \ldots, d_{I}\}$. These rules can be nonrandomized or randomized, in the sense of \cite{Ferguson67} pp. 24-25. An important aspect of our analysis is that we allow the decision maker to choose randomly from the set of decision rules $\mathcal{D}$ and we represent such a random choice by an element in the $I-1$ simplex:  
\[ \Delta(\mathcal{D}) =\Delta^{I-1}\equiv \left\{(p_1,...,p_I) \in \mathbb{R}^{I} \: \Bigg| \: \sum_{i=1}^I p_i = 1, p_i \geq 0\right\}.\]
It is well known that allowing the decision maker to choose randomly is usually to his advantage.\footnote{Consider a ``matching pennies'' game with two players, each with two actions: left and right. Suppose that column player gets $M$ when matched and $-M$ when unmatched. If neither player is allowed to choose actions randomly, the worst-case payoff obtained by the column player is $-M$ regardless of the action chosen. If the column player can randomize, but the row player cannot, the worst-case payoff for the column player if he chooses each action at random with probability 1/2 is zero.} Moreover, there are two additional reasons why we would like to allow for the possibility of \emph{randomization}. The first one is that in the main application we will consider in the paper (the \emph{site selection} problem described in Section \ref{sec:application}), the random choice of experimental sites is viewed as the default practice in applied work. The second reason is that, as we will explain in Section \ref{sec:approximate_solutions} (end of Remark \ref{remark:optimality}), allowing for random choice of actions can reduce the computational burden of selecting a good decision rule. 

A \emph{risk function} is used to summarize the performance of each decision rule $d_i \in \mathcal{D}$. This performance is contingent on the data generating process, which we parameterize by an element $\theta$ belonging to some space $\Theta$. Thus, we write the risk function of each decision rule $d \in \mathcal{D}$ as a mapping $R:\mathcal{D} \times \Theta \rightarrow \mathbb{R}$. We refer to $\theta$ as a parameter, and to $\Theta$ as the parameter space. We are particularly interested in the case in which $\Theta$ is an infinite set; for example, when $\Theta$ equals all of $\mathbb{R}^{d}$. We also want to allow for the possibility that each element in the parameter space is an infinite dimensional object (for example, when $\theta$ itself is a function). We impose the following assumption on the risk function:

\begin{assumption} \label{asn:A1_Risk}
    There exists a known constant $0<M<\infty$ such that for any $d \in \mathcal{D}$ and $\theta \in \Theta$, $0 \leq R(d,\theta) \leq M$.
\end{assumption}

In Section \ref{sec:illustrations} we explain how this assumption can be verified for each of the illustrative examples we consider. We view Assumption \ref{asn:A1_Risk} as a minimal regularity condition for the minimax problem to be well-behaved. We also note that the assumption holds if each of the $I$ decision rules under consideration has a finite worst-case risk. The algorithm requires this bound to be known. 

In a slight abuse of notation, we extend the original domain of the risk function---which was defined over decision rules in $\mathcal{D}$---to all possible random selections in $\Delta^{I-1}$. We do this by defining, for any $p \in \Delta^{I-1}$ and $\theta \in \Theta$, the function:  
\begin{equation} \label{eqn:risk_function}
R(p,\theta) \equiv \sum_{i=1}^I p_i R(d_i,\theta).
\end{equation}

We view a decision problem as a triplet $(\mathcal{D},\Theta,R(\cdot,\cdot))$ and we define the \emph{minimax value} of the decision problem as the scalar
\begin{equation} \label{eqn:minimax_value}
\Bar{v} \equiv \inf_{p \in \Delta^{I-1}} \sup_{\theta \in \Theta} R(p,\theta).
\end{equation}
A random selection $p^* \in \Delta^{I-1}$ is said to be a \emph{minimax} decision rule if
\begin{equation} \label{eqn:minimax_rule}
\sup_{\theta \in \Theta} R(p^*,\theta) = \Bar{v}.
\end{equation}
The use of the minimax criterion as a solution concept in statistical decision problems is traditional, dating back to \cite{Wald50}. 
\cite{manski2021econometrics} argues that the primary challenge in using the minimax criterion and \cite{Wald50}'s statistical decision theory is computational.

\subsection{Minimax solutions via convex programming} \label{subsection:convex_programming}
We first show that the minimax solution of the decision problems considered in this paper can be computed via convex programming. This observation is based on an analogous result in \cite{chamberlain2000econometric} by reversing the roles of Nature
and the decision maker; see Equation 5, p. 630, and the discussion therein.
The argument is as follows. For $p \in \Delta^{I-1}$, define the nonlinear function 
\begin{equation} \label{eqn:f_convex}
f(p) \equiv \sup_{\theta \in \Theta} R(p,\theta).   
\end{equation}
This function is the upper envelope---over all possible values in the parameter space---of the risk of $p$.

\begin{lemma} \label{lemma:Lemma_convex} Suppose Assumption \ref{asn:A1_Risk} holds. The function $f: \Delta^{I-1} \rightarrow \mathbb{R}$ is convex and Lipschitz continuous w.r.t. $\|\cdot\|_1$ (with constant at most $M$). Furthermore, fix an arbitrary $p_0 \in \Delta^{I-1}$. If there exists $\theta_0 \in \Theta$ such that $R(p_0,\theta_0) = f(p_0)$, then the vector $g_0$ in $\mathbb{R}^{I}$ given by
\begin{equation} \label{eqn:subgradient} 
g_0 \equiv (R(d_1,\theta_0), \ldots, R(d_I,\theta_0))^{\top}   
\end{equation}
is a subgradient of $f$ at $p_0$.\footnote{If $f: \Delta^{I-1} \rightarrow \mathbb{R}$ is convex, a vector $g_0$ is said to be a subgradient of $f$ at a
point $p_0$ if $f(p)\geq f(p_0) + g_0^{\top}(p - p_0), \forall p \in \Delta^{I-1}$. See \citet[][pp. 214-215]{rockafellarconvex}. }
\end{lemma}

\begin{proof}
The convexity of $f(\cdot)$ follows from \cite{chamberlain2000econometric}. The Lipschitz continuity follows from the definition of $f(\cdot)$. We provide a detailed proof in Appendix \ref{subsection:Proof_Lemma_convex}.
\end{proof}
Lemma \ref{lemma:Lemma_convex} shows that solving the minimax problem in \eqref{eqn:minimax_value} can be viewed as a nonlinear convex program over the $(I-1)$ simplex. We note that the connection to convex programming is helpful, but should not be viewed as a computational panacea. Evaluating the objective function of the convex program and its subgradient could remain computationally costly. 

To make sure that the subgradient in \eqref{eqn:subgradient} is well defined, we make the following assumption. 

\begin{assumption}  \label{asn:A2_oracle} 
For any $p \in \Delta^{I-1}$, there exists $\theta_p \in \Theta$---and an \emph{oracle} to find it---such that
\[\sum_{i=1}^I p_i R(d_i,\theta_p) = \sup_{\theta \in \Theta} \sum_{i=1}^I p_i R(d_i,\theta).\]
\end{assumption}

The assumption says that for any $p \in \Delta^{I-1}$ it is possible to find an element $\theta_{p}$ such that $R(p,\theta_p)=f(p)$. This means that there is an algorithm that can  i) evaluate the function $f(p)$ and ii) find a maximizer that evaluates to $f(p)$.

Assumption \ref{asn:A2_oracle} requires that the worst-case risk is attained and that there is an algorithm to find $\theta_{p}$. Later we discuss the extent to which Assumption \ref{asn:A2_oracle}  can be relaxed, by, for example, only requiring that we can get a $\delta$-approximation to $f(p)$. See Remark \ref{remark:relax_A3}. We note, however, that Assumption  \ref{asn:A2_oracle} is crucial for our algorithm and should not be taken for granted. Solving the optimization problem over $\theta$ could be computationally challenging depending on the structure of the risk function. Assumptions similar to our Assumption \ref{asn:A2_oracle} are not uncommon in the optimization literature and are usually presented as the existence of an oracle that can evaluate the objective function of interest (in our case, the worst-case risk of a particular $p \in \Delta^{I-1}$).\footnote{The computation of the best response in Assumption \ref{asn:A2_oracle} (for a minimax or maximin problem) can be feasible even when the
decision space or parameter space is large. For example, \cite{montielolea2026bounded} show how to access the oracle for the maximin problem associated to the minimax estimation of a  bounded normal mean; \cite{guggenberger2025numerical} show how to access the oracle in their leading application which features a large action space. Other examples of feasible oracles include  \cite{manski2023dosage},  in which choosing the optimal dosage for a medical treatment can be formulated as a linear program, and  \cite{gilles2020evasive} where a best-response oracle can be computed through a Hamilton--Jacobi equation. }

\section{Approximate Solutions for Minimax Problems} \label{sec:approximate_solutions}

A popular algorithmic approach for \emph{approximately} solving convex optimization problems (in particular, those of high dimensions) is to use the methods of mirror descent of  \cite{nemirovski1983problem}. It is known that the rate of convergence of mirror descent for convex problems in the simplex (as the one associated to our minimax problems) improves over regular subgradient descent \citep[Section 4.3]{bubeck2015convex}. 

This section starts out by presenting a formal definition of an \emph{approximate} minimax solution. Then, this section presents an off-the-shelf implementation of mirror subgradient descent \citep[Section 4.3, p. 301]{bubeck2015convex} that provably finds such an approximate solution (see Algorithm \ref{alg:mirror_descent}).

\subsection{$\epsilon$-Minimax Decision Rules} 

\begin{definition} \label{definition:epsilon_minimax}
[\cite{Ferguson67}, p. 33] A random selection $p_\epsilon^\star \in \Delta^{I-1}$ is an \emph{$\epsilon$-minimax} decision rule for the decision problem $(\mathcal{D},\Theta,R(\cdot,\cdot))$ if
    \[\sup_{\theta \in \Theta} R(p_\epsilon^\star,\theta) \leq \inf_{p \in \Delta^{I-1}} \sup_{\theta \in \Theta} R(p,\theta) + \epsilon = \Bar{v} + \epsilon.\]
\end{definition}
We note that the risk of an $\epsilon$-minimax decision rule is smaller---up to an additive factor of size $\epsilon$---than the worst-case risk of any other decision rule. That is:
    \[R(p_\epsilon^\star,\theta) \leq \sup_{\theta \in \Theta} R(p,\theta) + \epsilon, \: \forall \theta \in \Theta, \: \forall p \in \Delta^{I-1}.\]
The definition of a minimax decision rule further implies that
    \[\Bar{v} \leq \sup_{\theta \in \Theta} R(p_\epsilon^\star,\theta) \leq \Bar{v} + \epsilon.\]

\subsection{Mirror Subgradient Descent for finding $\epsilon$-Minimax Rules} 

We now show that a mirror subgradient descent for convex optimization can provably find an $\epsilon$-minimax solution. 

The pseudocode below describes a mirror subgradient descent routine for finding the minimum of \eqref{eqn:f_convex} over the simplex $\Delta(D)$.\footnote{The routine is taken from \cite{bubeck2015convex} (Section 4.3, p. 301), where the mirror map is chosen to be the negative entropy $\phi(x) = \sum_{i=1}^{n} x_i \log x_i$, the routine $\nabla \phi(x_{t+1}) = \nabla \phi(x_t) - \eta \nabla f(x_t)$ becomes $x_{t+1, i} = x_{t, i} \exp(- \eta \nabla f(x_t)_i)$. We simply adjust the notation to our problem.} 

\begin{algorithm}[h!]
\caption{Mirror Subgradient Descent, stopped after $T$ epochs.}
\label{alg:mirror_descent}
\begin{algorithmic}[1]
\State \textbf{Input:} Step-size $\eta > 0$ and number of epochs
       $T \in \mathbb{N}$.
\State Initialize $w_0 \in \mathbb{R}^I$ by setting
       $w_{i,0}=1$ for all $i \in \{1,\ldots,I\}$.
\For{$t=1,\ldots,T$}
    \State Compute $\phi_t \equiv \sum_{i=1}^I w_{i,t-1}$.
    \State For each $i \in \{1,\ldots,I\}$, compute
    \[
        p_{i,t} \equiv \frac{w_{i,t-1}}{\phi_t}.
    \]
    \State Find
    \[
        \theta_t \in \arg\sup_{\theta \in \Theta}
        \sum_{i=1}^I p_{i,t} R(d_i,\theta).
    \]
    \State Define the vector
    \[
        g_t \equiv
        \bigl(R(d_1,\theta_t),\ldots,R(d_I,\theta_t)\bigr)^\top.
    \]
    \State For each $i \in \{1,\ldots,I\}$, update
    \[
        w_{i,t} \equiv w_{i,t-1}\exp(-\eta g_{i,t}).
    \]
\EndFor
\State \textbf{Output:} $\frac{1}{T}\sum_{t=1}^T p_t$.
\end{algorithmic}
\end{algorithm}

As shown in Lemma \ref{lemma:Lemma_convex}, the gradient vector $g_{t}$ collects the risk associated to each decision rule at $\theta_t$ (the point in the parameter space associated to the worst-case performance of $p_t$). The mirror descent update in Algorithm \ref{alg:mirror_descent} is intuitive: decision rules with high risk at $\theta_t$ are used less frequently in the following round. Algorithm \ref{alg:mirror_descent} assumes that the subgradient, $g_t$, is known. However, there are versions of the algorithm that replace $g_t$ by an unbiased estimator; see our Remark \ref{remark:SMD} about \emph{stochastic mirror descent}  and Chapter 6 in \cite{bubeck2015convex}.   

The mirror subgradient routine in Algorithm \ref{alg:mirror_descent} is also known in the computer science literature as the \emph{Hedge Algorithm} (a particular case of the \emph{Multiplicative Weights} update method).\footnote{The Multiplicative Weights update method is a popular algorithm in computer science that has found different applications in machine learning; see \cite{arora2012multiplicative}. The specific version of the Multiplicative Weights algorithm used in this paper uses an exponential function of each of the coordinates of the gradient to update the weights and is known as the Hedge algorithm. See Section 2.1 in \cite{arora2012multiplicative}.} The Hedge algorithm is used in algorithmic game theory as a practical tool to find approximate solutions of two-person zero-sum games. Importantly, the mirror subgradient descent routine typically uses $(1/T)\sum_{t=1}^{T} p_{t}$ (and not the last $p$ obtained in the iteration) as the approximate minimizer.\footnote{Averaging the trajectories of a gradient-descent routine is commonly referred to as Polyak-Ruppert averaging. See \cite{ruppert1988efficient} and \cite{polyak1992acceleration}. See also \cite{forneron2024estimation} for a discussion of Polyak-Ruppert averaging in the context of estimation and inference by stochastic optimization of nonlinear econometric models.} 

We now show that if we set the step size to $\eta \equiv \epsilon/M^2$ and stop the routine after $T = \lceil 2 M^2\ln(I) / \epsilon^2 \rceil$ epochs, the mirror subgradient descent routine in Algorithm \ref{alg:mirror_descent} can provably find an $\epsilon$-minimax solution (in the sense of Definition \ref{definition:epsilon_minimax}). Our main result---which follows directly from the properties of mirror subgradient descent in convex optimization problems---is the following:  

\begin{thm} \label{thm:MWU} 
Suppose Assumptions \ref{asn:A1_Risk}-\ref{asn:A2_oracle} hold. If $\epsilon \leq M$, $\eta \equiv \epsilon/M^2$, and $T \equiv \lceil2 M^2  \ln(I)/\epsilon^2 \rceil$, 
then the random choice of decision rules that assigns probability
\[ p^{\epsilon}_{i} \equiv \frac{1}{T}\sum_{t=1}^T p_{i,t}\]
to each decision rule $d_i$---where $p_{t}$ corresponds to the $t$-th iteration of the mirror descent routine in Algorithm \ref{alg:mirror_descent}---is $\epsilon$-minimax in the sense of Definition \ref{definition:epsilon_minimax}. 
\end{thm}

\begin{proof} We present two different proofs of Theorem \ref{thm:MWU}. First, in Appendix \ref{subsection:Proof_Thm_1_bubeck} we apply a well-known result from the convex optimization literature that shows that mirror subgradient descent provably finds an $\epsilon$-approximate minimizer of a convex function; in particular, we verify the conditions of Theorem 4.2 in \cite{bubeck2015convex}. Second, in Appendix \ref{subsection:Proof_Thm_1} we adapt (and extend) the results in \cite{arora2012multiplicative} concerning the use of the Hedge algorithm to approximate the minimax solution of two-player zero-sum games where both players have finitely many pure strategies. In adapting and extending their results, we improve the number of epochs by a factor of two, and match the results in Theorem 4.2 in \cite{bubeck2015convex}.  
\end{proof}

Theorem \ref{thm:MWU} presents a concrete computational strategy to approximately solve the statistical decision problems considered in this paper. The only tuning parameter that needs to be chosen is $\epsilon$, which controls the approximation error. We note that in cases where it is difficult to commit to a value of $\epsilon$ explicitly, one can solve for the value of $\epsilon$ if there is a specific target for the runtime of the algorithm and we know the time it takes for each iteration to run.

We make some important remarks about Theorem \ref{thm:MWU}.

\begin{remark}\label{remark:optimality}[\emph{Optimality of Algorithm \ref{alg:mirror_descent}}] There is a sense in which Algorithm \ref{alg:mirror_descent} is essentially the best first-order iterative algorithm for minimizing a convex, Lipschitz function over the simplex; see Proposition 4.2 in \cite{ben2001ordered}. We briefly review the notation in \cite{ben2001ordered} and summarize their findings.

Let $\mathcal{F}(M,I)$ denote the collection of all minimization problems of a convex, Lipschitz function $f$ (with respect to $\|\cdot\|_1$ and with constant at most $M$) over the $I-1$ simplex. Since the minimization problem is indexed entirely by the function $f$, we denote the elements of $\mathcal{F}(M,I)$ succinctly as $f$. Let $\partial f(p)$ denote the subdifferential of $f$ (the set containing all subgradients) at $p$. Let $A$ be a \emph{first-order} iterative algorithm that successively generates points $p_t (A,f) \in \Delta^{I-1}$ and approximate solutions $p^{t}(A,f)$. We restrict the class of algorithms by requiring both $p_t$ and $p^{t}$ to be deterministic functions of first-order information about $f$; namely the history of evaluations of $f$ and its subdifferential: $\{ f(p_s), \partial f(p_s)  \}_{s=1}^{t-1}$. For the starting search point or \emph{initial condition}, $p_1$, we require it to be chosen independently of the function $f$. We denote the class of deterministic, iterative, first-order algorithms as $\mathcal{A}$. Given a tolerance $\epsilon$, define the \emph{complexity} of the class of optimization problems $\mathcal{F}(M,I)$ with respect to algorithm $A$ as the function
\[ \textrm{Complexity}_{A}(\epsilon; \mathcal{F}(M,I)) \equiv \inf\{ T \in \mathbb{N} \: | \: f(p^t(A,f)) - \inf_{p \in \Delta^{I-1}} f(p) \leq \epsilon, \quad \forall t \geq T, f \in \mathcal{F}(M,I) \}. \]
This is the smallest number of calls needed by algorithm $A$ to generate an $\epsilon$-approximate solution for any convex optimization problem over the simplex. Define the complexity of the family of optimization problems in $\mathcal{F}(M,I)$ as 
\[ \textrm{Complexity} (\epsilon; \mathcal{F}(M,I)) \equiv \inf_{A \in \mathcal{A}} \textrm{Complexity}_{A}(\epsilon; \mathcal{F}(M,I)).  \]
Proposition 4.2 in \cite{ben2001ordered} shows that 
\[\textrm{Complexity} (\epsilon; \mathcal{F}(M,I)) \geq O(1) \min\{ M^2 / \epsilon^2 , I \}.  \]
Therefore, the smallest number of calls to the objective function and its subgradient required by \emph{any} iterative, first-order algorithm for (convex) optimization over the $(I-1)$ simplex of a Lipschitz function with constant at most $M$ (with respect to $\| \cdot \|_1$ norm) is $O(1)M^2/ \epsilon^2$, provided $\epsilon \geq M / \sqrt{I}$. Thus, Algorithm \ref{alg:mirror_descent}, and the suggested number of iterations in Theorem \ref{thm:MWU}, are optimal up to the logarithmic factor $\ln(I)$.\footnote{There exist different results in the computer science literature providing lower bounds for the \emph{regret} of the Multiplicative Weights update method in problems where a decision maker chooses randomly among $I$ alternatives; see Section 4 in \cite{arora2012multiplicative} and also \cite{lowerboundstight}. We note, however, that these regret bounds do not speak directly to the question of whether there exists another iterative algorithm for convex optimization of an $M$-Lipschitz function over the simplex that could find an $\epsilon$-minimizer in fewer epochs than the Hedge Algorithm.} \\
Finally, ignoring randomized strategies, if one were interested in computing the minimax optimal rule among $\{d_{1},\ldots, d_{I}\}$, one would need $I$ calls to the oracle (one for computing the worst-case performance of each rule). Surprisingly, the $\epsilon$-minimax solution among randomized rules could be done in fewer calls ($\lceil 2M^2 \ln (I)/\epsilon^2\rceil$). When $I$ is large, the difference could be substantial.
\end{remark}

\begin{remark}[\emph{Direct and Indirect Approximate Least-favorable distribution}] \label{remark:LFD}
The statistical decision problems we study in this paper can be interpreted as the following two-player zero-sum game: the two players are 1) the statistician who has pure strategies $\mathcal{D} = \{ d_1, d_2, ..., d_I\}$, and 2) ``nature'',  whose set of pure strategies is given by the parameter space $\Theta$. The payoff function is $R(d_i, \theta)$. In the mixed extension of the game, in each round, the statistician first chooses a mixed strategy $p_t \in \Delta^{I-1}$, and then nature responds with a choice $\theta_t$. Surprisingly, Algorithm~\ref{alg:mirror_descent} not only outputs (provably) an approximate minimax solution for the statistician, but also \emph{indirectly} gives (provably) an approximate maximin solution for nature. In particular, the empirical distribution of the sequence of  nature's best responses, $\{\theta_{t}\}_{t=1}^{T}$, 
is an \emph{$\epsilon$-maximin} strategy for nature. See Appendix~\ref{sec:S_games} for a detailed explanation. More generally, it is worth noting that it is straightforward to modify Algorithm \ref{alg:mirror_descent} to \emph{directly} find maximin solutions to statistical decision problems in which nature has finitely many pure strategies $\{\theta_1, \ldots, \theta_{I}\}$, even when the space of decision rules for the statistician is unrestricted (instead of doing mirror descent, we simply do mirror ascent).\footnote{While there are clearly many instances of maximin problems in which the parameter space for nature has finitely many elements (see, for example, \cite{gilles2020evasive}), maximin problems in statistical decision theory typically feature an infinite parameter space. While this parameter space can always be discretized (for example, see \cite{chamberlain2000econometric}, \cite{hartline2024subgame,guo2025algorithmic}), the justification for such a discretization is very different from the one we use to focus on minimax problems with finitely many decision rules. As mentioned in the introduction, our motivation to consider finitely many decision rules reflects the practical consideration that one is typically only interested in the performance of relatively simple rules that could be used in  practice.} The best response for the statistician is obtained via Bayes risk minimization. The $t$-th epoch update for the probabilities attached to each element of the set $\{\theta_1, \ldots, \theta_{I}\}$---denoted $\pi_{i,t}$---takes the form: 
\[
    \pi_{i,t} \propto \pi_{i,t-1} \cdot \exp(\eta \cdot \tilde{g}_{i,{t}}),
    \]
where  
\[ \tilde{g}_{t} \equiv (R(d_{t-1}, \theta_1), \ldots, R(d_{t-1}, \theta_I))^\top, \]
and $d_{t-1}$ is a Bayes rule associated to the prior $(\pi_{i,t-1}, \ldots, \pi_{I,t-1})$. The provably $\epsilon$-maximin solution for nature (or approximate least-favorable distribution) is
\[ \frac{1}{T} \sum_{t=1}^{T} \pi_{t}. \]
An $\epsilon$-minimax decision rule can be recovered from the algorithm that finds the approximate least-favorable distribution by considering the randomized decision rule that attaches probability $1/T$ to the decision rule $d_{t}$. Note that simply reporting a Bayes rule associated to the approximate least-favorable distribution needs not be $\epsilon$-minimax.\footnote{This is because not all Bayes rules associated to the least-favorable distribution are necessarily minimax optimal. See also Section \ref{subsection:MMR} for additional discussions on this point.  } Other recent algorithms for directly solving a certain class of maximin problems can be found in \cite{balter2024solving}. 
\end{remark}

\begin{remark}[\emph{Approximate evaluation of $f(p)$}]
\label{remark:relax_A3} 
It is possible to extend the results of Theorem \ref{thm:MWU} to the case in which $\theta_t$ is not the exact solution of the problem in \eqref{eqn:f_convex}, but an approximate one. More precisely, consider $\theta_t^\delta$ such that
\begin{equation}\label{eq:theta.approximate}
\left(\sup_{\theta \in \Theta} \sum_{i=1}^I p_{i,t}R(d_i,\theta)\right) - \delta \leq \sum_{i=1}^I p_{i,t}R(d_i,\theta_t^\delta) \leq \sup_{\theta \in \Theta} \sum_{i=1}^I p_{i,t}R(d_i,\theta).   
\end{equation}
This extension can be (roughly) completed by slightly adjusting the proof of Theorem \ref{thm:MWU} in Appendix \ref{subsection:Proof_Thm_1}. In Appendix \ref{sec:approximate}, we  show that choosing $T$ as we have done gives an $\epsilon + \delta$ approximation. 
\end{remark} 

\begin{remark}[\emph{Minimax problems with infinitely many decision rules}] \label{remark:infiniteD} The typical minimax problem in statistical decision theory takes the form 
\begin{equation} \label{eqn:general_minimax}
\bar{v}^*  \equiv \inf_{d \in \mathcal{D}^*} \sup_{\theta \in \Theta} R(d,\theta), 
\end{equation}
where $\mathcal{D}^*$ is usually the space of all \emph{randomized} decision rules. While there are cases in which the error from the discretization can be derived explicitly---as we explain in Appendix \ref{app:decision-discretization}, using the illustrative example of Section \ref{subsection:MMR}---we think whenever the problem of interest has infinitely many decision rules it is more natural to use the $\epsilon$-minimax/maximin solutions to provide upper/lower bounds on the minimax value of the problem of interest. Suppose the $I$ decision rules in $\mathcal{D}$ are nonrandomized. Then, given  $p^\epsilon$, there is a randomized rule  $d^{\epsilon} \in \mathcal{D}^*$ such that $R(d^\epsilon,\theta) = R(p^\epsilon,\theta)$ for every $\theta$.\footnote{For any $x \in \mathcal{X}$, simply take $d^{\epsilon}(x)$ to be the discrete distribution $p^\epsilon$ over the actions $(d_1(x),\ldots, d_I(x))$.} This means that 
\begin{equation} \label{eqn:upper_bound_infinite_D}
\bar{v}^* \leq \sup_{\theta \in \Theta} R(p^\epsilon,\theta) = f(p^\epsilon).
\end{equation}
Similarly, let $q^\epsilon$ denote the least-favorable distribution obtained from Algorithm \ref{alg:mirror_descent}, as explained in Remark \ref{remark:LFD}. The standard relation between average and maximum risk implies
\begin{equation} \label{eqn:lower_bound_infinite_D}
\inf_{d \in \mathcal{D}^*} \mathbb{E}_{q^\epsilon}[R(d,\theta)] \leq \bar{v}^*.
\end{equation}
This means that the output of Algorithm \ref{alg:mirror_descent} can be used to upper and lower bound the minimax value in \eqref{eqn:general_minimax}. Moreover, if we define $\epsilon^* \equiv f(p^{\epsilon})- \inf_{d \in \mathcal{D}^*} \mathbb{E}_{q^\epsilon}[R(d,\theta)] \geq 0$, then $p^\epsilon$ is an $\epsilon^*$-minimax decision rule for the problem with infinitely many decision rules. This is simply because 
\[ \bar{v}^* \leq f(p^\epsilon) \leq \left( \bar{v}^* -\inf_{d \in \mathcal{D}^*} \mathbb{E}_{q^\epsilon} [R(d,\theta)]\right)   + f(p^\epsilon) = \bar{v}^* + \epsilon^*.  \]
This means that Algorithm \ref{alg:mirror_descent}---when applied to finitely many decision rules---generates an $\epsilon^*$-minimax decision rule for the problem in \eqref{eqn:general_minimax}. The caveat is that $\epsilon^*$ is determined ex-post, and not ex-ante as in Theorem \ref{thm:MWU}. It is also worth mentioning that there are versions of the Hedge algorithm for problems where both players have infinitely many actions; see for example \cite{krichene2015hedge}.\footnote{Broadly speaking, their results can be used to show that, for any target $\epsilon$, there exists a large enough number of iterations that guarantee that the output of the Hedge algorithm is an $\epsilon$-approximate solution. From the perspective of implementation, the required number of iterations can be shown to depend on the Kullback-Leibler divergence between the initial condition and the exact solution to the problem. Since the exact solution is unknown, it becomes harder to give explicit recommendations as those we provided in Theorem \ref{thm:MWU}.} 
\end{remark}

\begin{remark}[\emph{Stopping the algorithm before our suggested $T$ epochs}]\label{remark:early_stopping} The recommended number of epochs in Theorem \ref{thm:MWU} provably finds an $\epsilon$-minimax solution for \emph{any} risk function that satisfies Assumptions \ref{asn:A1_Risk} and \ref{asn:A2_oracle}. We note, however, that for a particular risk function at hand it is possible to stop the algorithm before $T \equiv \lceil2 M^2  \ln(I)/\epsilon^2 \rceil$ rounds. This means that, in any given application, our suggested number of rounds is best viewed as an upper bound on the number of iterations needed to provably generate an $\epsilon$-minimax solution. The main observation is that---using the same arguments we used to derive \eqref{eqn:upper_bound_infinite_D} and \eqref{eqn:lower_bound_infinite_D}---we can show that
\[ \inf_{p \in \Delta^{I-1}} \mathbb{E}_{q^\epsilon} [ R(p,\theta)] \leq \bar{v} \leq \sup_{\theta \in \Theta} R(p^{\epsilon},\theta).      \]
Consequently, in each iteration of Algorithm \ref{alg:mirror_descent} one could check the gap between the upper and the lower bound (but evaluated at the candidate values of $q^\epsilon$ and $p^\epsilon$ at round $\widetilde{T} \leq T$) and stop whenever the gap is smaller than $\epsilon$. Note that to evaluate the lower bound, it suffices to keep track of 
\[ \frac{1}{\widetilde{T}} \sum_{t=1}^{\widetilde{T}} R(d_i,\theta_t),\]
for any round $\widetilde{T} \leq T$, and for each decision rule $d_i$. The upper bound can be evaluated directly, or we can use an upper bound based on the output of the algorithm; see Appendix \ref{remark:alternative-value}. We also implemented this early-stopping algorithm in the treatment-choice example of Section \ref{subsection:MMR}; see Table \ref{tab:comparison.1} in Appendix \ref{app:gp} for details. 
\end{remark}

\begin{remark}[\emph{Finite $\Theta$}] 
When $\Theta$ has $J$ elements, obtaining an exact minimax solution could be done via a linear program \citep[Section III.1, p. 36]{dantzig1951proof, adler2013equivalence,owen2013game}. The computational cost of using the fastest solver for linear programs can be shown to be of order $(1+J+I)^{2.25}$ time.\footnote{\cite{jiang2020faster} show that the fastest known LP solver for general (dense) linear programs can solve such a program in an order of approximate $(1+I+J)^{2.055}$ time, conditional on matrix multiplication exponents of 2. The recent results released by Open AI suggest that the matrix multiplication exponent is 2.25, which then implies---by the results of \cite{jiang2020faster}---that the fastest known LP solver requires an order of approximate $(1+I+J)^{2.25}$ time.} We note that Algorithm \ref{alg:mirror_descent} makes $\lceil 2 M^2 \ln(I)/\epsilon^2\rceil$ calls to nature's oracle. Suppose that the runtime of the oracle is $r(I,J)$. In each round, the algorithm evaluates the risk of the $I$ actions available to the decision maker. Thus, the runtime of the algorithm is of order
\[M^2 I \ln (I) r(I,J) / \epsilon^2.\] 
If the calls to the oracle that computes nature's best response are not expensive, and if $M/\epsilon^2$ is not too large, the time needed in order to compute the approximate solution to the minimax problem could be smaller than that time needed to obtain the exact solution. We also note that when $\Theta$ has $J$ elements, there might also be better algorithms to find an $\epsilon$-minimax decision rule; see, for example, the Saddle-Point Mirror Prox algorithm discussed in Section 5.2, p. 317 of \cite{bubeck2015convex}. In general, if one is willing to make more assumptions beyond ours, there might be better algorithms for solving the minimax problems of interest.
\end{remark}

Lastly, we mention there is also recent interest in approximating the solution of minimax problems in which the strategies for both the statistician and nature are parameterized via neural networks, with weights that are updated iteratively using versions of what is called \emph{subgradient ascent-descent}; see the recent work of \cite{luedtke2020learning} and also \cite{luedtke2021adversarial}. These algorithms where two players use subgradient descent are similar to the approaches used when optimizing Generative Adversarial Networks (GANs); see, for example, \cite{kaji2023adversarial}. These subgradient ascent-descent algorithms are also commonly used to approximate the equilibrium of two-person zero-sum games by invoking simultaneous no-regret dynamics; see, for example, Section 3.1 in \cite{lewis2018adversarial} and the references therein. Convergence rates for these subgradient ascent-descent algorithms, as well as some performance guarantees for a finite number of iterations,  are available under some conditions. It is known, however, that the (approximate) stationary points of these gradient ascent-descent algorithms are not necessarily $\epsilon$-minimax strategies. Instead, they are close to what the literature refers to as \emph{local min-max} solutions; see the seminal work of \cite{daskalakis2021complexity}. As we discuss in the conclusion, it would be interesting to further explore the differences between $\epsilon$-minimax strategies and the notion of a local min-max point.

\section{Illustrative Examples} \label{sec:illustrations}

\subsection{$\epsilon$-Minimax Regret Treatment Choice with Partial Identification} \label{subsection:MMR}
Consider the following example taken from  \cite{stoye2012minimax} and \cite{yata2021}. A policy maker uses experimental data to decide whether to implement a new policy in a target population of interest. The treatment effect of action $a=1$ is $\mu^*\in\mathbb{R}$, while the effect of action $a=0$ is normalized to be equal to $0$. Thus, the policy maker's expected payoff equals $W(a,\mu^*) \equiv a \cdot \mu^*$. 

The data available to the policy maker is an estimated treatment effect, $\hat{\mu}$, for the experimental population. The policy maker assumes that
\begin{equation} \label{eq:stoye_data}
\hat{\mu} \sim N(\mu,\sigma^2),
\end{equation}
where $\sigma>0$ is known and where $\mu \in\mathbb{R}$ is the true effect of the policy in the population where the experiment was conducted.  The policy maker is concerned about the external validity of the experiment at hand. This is captured by allowing the effect of the policy in the experimental population ($\mu$) to be different from the effect in the target population ($\mu^*$). The policy maker is willing to work under the assumption that $\left\vert \mu^*-\mu \right\vert \leq k$ for some known $k \geq 0$. In this example, $\theta=(\mu,\mu^*)^{\top}$ and $\Theta \equiv \{(\mu,\mu^*) \in \mathbb{R}^2 \: | \: | \mu - \mu^* | \leq k \}\subseteq \mathbb{R}^2$.

A decision rule for the policy maker is a mapping $d: \mathbb{R} \rightarrow[0,1]$ from the observed experimental data \eqref{eq:stoye_data} to an action $a\in[0,1]$. The action is interpreted as the fraction of the target population that will be treated. Consider the regret loss associated to $W(a,\mu^*)$ given by $L(a,\theta) \equiv \mu^*[ \mathbf{1}\{ \mu^* \geq 0\} - a ]$. Define the risk function
\[ R(d,\theta) \equiv \mathbb{E}_{\theta}[L(d,\theta)].   \]

{\scshape Exact Minimax Regret Solution Over all Decision Rules:} Let $\mathcal{D}^*$ denote the set of all decision rules.  \cite{stoye2012minimax} derived a solution to the minimax (regret) problem 
\begin{equation} \label{eqn:mmr-stoye}
\inf_{d \in \mathcal{D}^*} \sup_{\theta \in \Theta} R(d,\theta),   
\end{equation}
as a function of $(\sigma^2,k)$. \cite{stoye2012minimax} showed that when $k \geq \sqrt{\pi/2} \sigma$, Equation $\eqref{eqn:mmr-stoye}$ equals $k/2$. \cite{MQS2023decision} further showed that, when $k > \sqrt{\pi/2} \sigma$, there are infinitely many minimax-regret optimal rules. One such solution takes the form:
\[d_{MQS}^\star(\hat{\mu}) = \begin{cases}
    0, & \hat{\mu} < -\rho^\star, \\
    \frac{\hat{\mu} + \rho^\star}{2 \rho^\star}, & -\rho^* \leq \hat{\mu} \leq \rho^*\\
    1, & \hat{\mu} > \rho^*,
\end{cases},\]
where $\rho^*<k$ is the unique strictly positive solution to
\begin{eqnarray}
\left( \frac{1}{2 k } \right) \rho^{*} -\frac{1}{2} + \Phi\left(-\frac{\rho^{*}}{\sigma}\right) =0,
\label{eq:rho.star.main.proof} 
\end{eqnarray} 
see Theorem 3 in \cite{MQS2023decision}.

{\scshape Approximate Minimax Regret Solution over a Class of Threshold Rules:} Suppose that instead of considering all decision rules, we focus on a class $\mathcal{D} \subset \mathcal{D}^*$ that contains only ``threshold'' rules; that is, decision rules of the form 
\[ d_i(\hat{\mu}) \equiv \mathbf{1}\{ \hat{\mu} \geq c_i\}, \]
where $c_i \in \mathbb{R}$. For concreteness, we consider $I=501$ different values for $c_i$ equally spaced in the interval $[-k,k]$, including the endpoints. One motivation for using threshold rules is that, while it is always theoretically interesting to look for the \emph{best} overall decision rule, there are situations in which it is equally desirable to  \emph{``evaluate the performance of relatively simple statistical decision functions that researchers use in practice''} \citep{manski2024comprehensive} and choose optimally among them. Threshold rules seem reasonable for this problem. For example, if one observes a realization of $\widehat{\mu} \geq k$, any of these rules would recommend implementing the policy at scale. 

Algebra shows that, in this example, the largest worst-case risk among all threshold rules in $\mathcal{D}$ is bounded above by $M \equiv \sigma \max_{x \geq 0} x \Phi \left( (2k/\sigma) - x \right)$, where $\Phi \left( \cdot \right)$ denotes the standard normal c.d.f.\footnote{The formula corresponds to the worst-case risk of the rule that uses the threshold $c_{i}=k$ (or -$k$).} Since the expected loss is nonnegative, Assumption \ref{asn:A1_Risk} is satisfied.  

We can also show that, for a given $p \in \Delta^{I-1}$, the values $(\mu,\mu^*) \in \Theta$ that verify Assumption \ref{asn:A2_oracle} can be obtained by solving three optimization problems. Define the parameter $\mu^*_{+}$ to solve\[ \max_{\mu^* \geq 0} \mu^* \left( \sum_{i=1}^{I} p_{i} \Phi \left( \frac{c_i-\mu^*}{\sigma} + \frac{k}{\sigma} \right)  \right), \]
and $\mu_{+} \equiv \mu_{+}^*-k$. Define the parameter $\mu^*_{-}$ to be the solution of the problem
\[ \max_{\mu^* \leq 0} -\mu^* \left( \sum_{i=1}^{I} p_{i} \Phi \left( \frac{\mu^*-c_i}{\sigma} + \frac{k}{\sigma} \right)  \right), \]
and $\mu_{-}= \mu^*_{-}+k$. Set $\theta_{p}$ to be the maximizer of
\[ \{ R(p,\mu_{+},\mu^*_{+}),  R(p,\mu_{-},\mu^*_{-}) \}. \]

Since we have verified Assumptions \ref{asn:A1_Risk} and \ref{asn:A2_oracle}, we apply Algorithm \ref{alg:mirror_descent}. We consider the case in which $\sigma = 1$ and $k=3$. The value of the bound $M$ is $M \approx 4.2484$. Since we know that the value of the problem in Equation \ref{eqn:mmr-stoye} is $3/2$ (under the parameters we have chosen), we can set $\epsilon=.15$ (that is, we are willing to tolerate 10\% relative error). We later discuss how to pick $\epsilon$ in more realistic problems in which there is no information about the minimax value. The number of epochs in Theorem \ref{thm:MWU} then becomes
\[ T = \lceil 2 M^2 \ln (I)/ \epsilon^2 \rceil =9{,}974. \]
The runtime of Algorithm \ref{alg:mirror_descent} is about 238.98 seconds (MATLAB R2023a on an Intel Core i9-13900H). Figure \ref{fig:epsilon_minimax} presents a comparison of $d^*_{MQS}$ and the $\epsilon$-minimax rule. The worst-case risk of the approximate minimax rule is about 1.500158. In Appendix \ref{app:decision-discretization} we also quantify the effect of discretization and show that it is small given our grid. 

\begin{figure}[htbp]
\centering
\includegraphics[width=.485\linewidth]{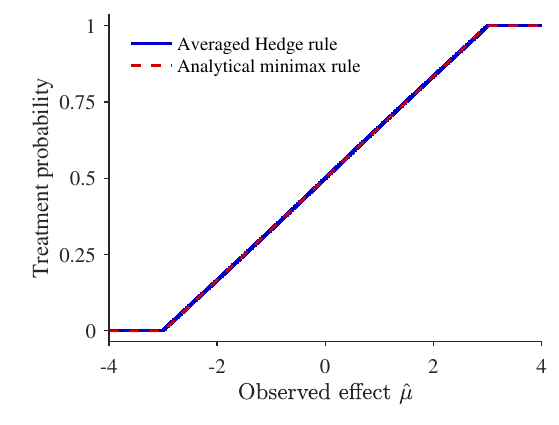}
\caption{Treatment probabilities for the averaged Hedge rule (solid blue) and
the analytical minimax rule (dashed red). The two curves are almost the same.}
\label{fig:epsilon_minimax}
\end{figure}

{\scshape Approximate Maximin Solution with a Discretized Parameter Space:} Suppose now that instead of obtaining an approximate minimax solution by discretizing the space of threshold rules, we discretize the parameter space $\Theta \equiv \{(\mu,\mu^*) \in \mathbb{R}^2 \: | \: | \mu - \mu^* | \leq k \}\subseteq \mathbb{R}^2$. 

While there are many different ways of discretizing this space, for concreteness, we consider an equally spaced grid of $501$ points for $\mu$ in $[-k, k]$, including the endpoints. Thus, the grid includes $\mu=0$. For each element $\mu$ in the grid, we consider an equally-spaced grid of $101$ points for $\mu^* \in [\mu-k,\mu+k]$. This gives a total of $I=501\times 101 = 50{,}601$ candidate support points for the least-favorable distribution. 

We then apply the mirror ascent algorithm described in Remark \ref{remark:LFD}. We assume that the space of decision rules is $\Delta(\mathcal{D})$, where $\mathcal{D}$ contains threshold rules of form $\mathbf{1}\{\hat{\mu}\geq c\}$ for  $c\in \{c_1,...c_{501}\}$ as used in mirror descent.\footnote{One can also consider the space of all decision rules for the mirror ascent routine. In this case,  the best response for the statistician---which verifies the required form of Assumption \ref{asn:A2_oracle}---is the rule that recommends treatment whenever the posterior mean of $\mu^*$ is larger than 0. The bottleneck is i) efficient computation of this rule, which in general may no longer be monotone in $\hat{\mu}$ and ii) evaluation of its risk function over the grid.}  We set again $\epsilon=.15$. The number of epochs in Theorem \ref{thm:MWU} then becomes
\[ T = \lceil 2 M^2 \ln (I)/ \epsilon^2 \rceil =17{,}378. \]
The runtime of Algorithm \ref{alg:mirror_descent} is about 14.76 seconds (MATLAB R2023a on an Intel Core i9-13900H). Figure \ref{fig:ma_prior}  reports the $\epsilon$-least-favorable distribution. For the problem under consideration, the exact least-favorable distribution is known to be the uniform distribution supported on the points $(\mu,\mu^*)$ given by $(0,k)$ and $(0,-k)$ (see, e.g., \citealt[][Proposition 7]{stoye2012minimax}). The approximate maximin value of this problem---that is, the Bayes risk of the candidate least-favorable distribution is $1.451185$, within the prescribed tolerance $\epsilon=0.15$. This verifies that the suggested mirror ascent routine did in fact produce a valid approximately least-favorable distribution.

\begin{figure}[htbp]
    \centering
    \includegraphics[width=0.55\linewidth]{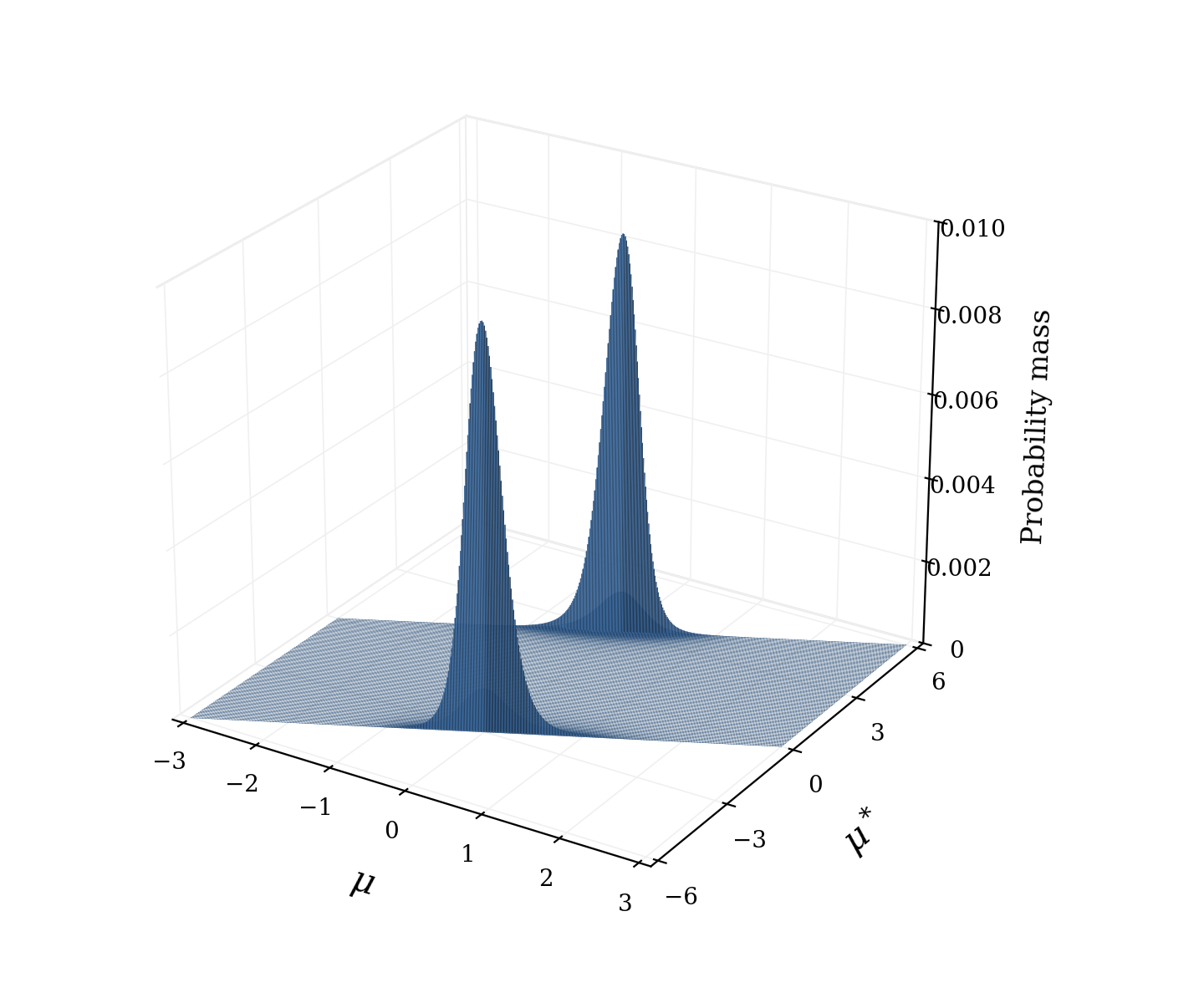}
    \caption{Averaged prior after $17{,}378$ rounds of mirror ascent.
    Each vertical stem gives the probability mass that $\bar{\pi}_T$
    assigns to $(\mu,\mu^*)$.}
    \label{fig:ma_prior}
\end{figure}

With this $\epsilon$-least-favorable distribution, it is tempting to proceed as in \cite{chamberlain2000econometric} and consider the corresponding Bayes rule as a potential candidate to serve as an $\epsilon$-minimax rule.  Unfortunately, in general, the Bayes rule with respect to a provable $\epsilon$-least-favorable prior can fail to deliver an $\epsilon$-minimax rule. 

Indeed, in this example, the Bayes rule corresponding to the $\epsilon$-least-favorable distribution obtained via mirror ascent is a threshold rule of the form $\mathbf 1[\hat{\mu}\geq 0]$. The worst-case risk of this threshold rule over all of $\Theta$ equals $1.74419$, which is more than $\epsilon=.15$ above the minimax value of $k/2=1.5$.\footnote{It is also known that even the best threshold rule (in the minimax sense) could have worst-case risk significantly larger than the minimax value; see \cite{MQS2023decision}. }  To get some intuition for this result, note that under the problem's exact least favorable distribution \emph{every} decision rule is a Bayes rule. However, in this problem, not every decision rule is minimax optimal. More generally, it is known that when a least-favorable distribution exists and the minimax theorem holds,  the minimax optimal rule has to be a Bayes rule with respect to the least favorable prior \cite[][p. 59]{Ferguson67}. But it is \emph{not} true that every Bayes rule with respect to the least favorable prior must be minimax optimal.\footnote{Denote by $d_{\pi_{\epsilon}}$ the Bayes rule with respect to the approximate least favorable distribution derived using the mirror ascent routine. An obvious sufficient condition that guarantees  $d_{\pi_{\epsilon}}$ to be $\epsilon$-minimax in the sense of Definition \ref{definition:epsilon_minimax} is that $\sup_{\theta\in\Theta}R(d_{\pi_{\epsilon}},\theta)\leq r(\pi_{\epsilon})+\epsilon$, where $r(\pi_{\epsilon})$ refers to the optimized average risk under the approximately least-favorable distribution $\pi_{\epsilon}$. Another possibility, explained in \cite{chamberlain2000econometric}, is to check whether the risk is approximately maximal at every support of the candidate least favorable distribution. }

Nonetheless, it is still possible to make use of the mirror ascent routine to \emph{indirectly} generate a provable $\epsilon$-minimax rule. Our theoretical results---see Remark \ref{remark:LFD} and Appendix \ref{sec:maximin.from.convex.routine}---show that the randomized rule that places equal probability on each Bayes rule obtained at each iteration of the algorithm is provably $\epsilon$-minimax. We implemented the mirror ascent routine in our example; see Table \ref{tab:comparison.1} in Appendix  \ref{app:gp} for details. In Appendix \ref{app:decision-discretization} we also quantify the effect of discretization and show that it is small given our grid.

\subsection{$\epsilon$-Robust Bayes Treatment Choice with Partial Identification}

Consider the same example as in Section \ref{subsection:MMR}, but instead of focusing on minimax-regret optimality as in \cite{stoye2012minimax}, we are interested in computing ex-ante Robust Bayes rules as in \cite{aradillas2024robust}. 

Let $\pi$ be a prior over $(\mu,\mu^\star)$. We are interested in obtaining the rule that minimizes worst-case expected risk over the class of priors suggested by \cite{giacomini2021robust}. We will denote this class of priors by $\Gamma$. Broadly speaking, the priors in this class fix a marginal prior over $\mu$, but allow for arbitrary priors over $\mu^*| \mu$ (as long as the joint distribution over $(\mu,\mu^*)$ is supported on $\Theta$). For this example, we will first consider the ``two-point prior'' for $\mu$ analyzed in \cite{aradillas2024robust}. That is, we assume that the prior of $\mu$ is supported on the set $\mathcal{M} = \{-\Bar{\mu},\Bar{\mu}\}$. We first assume that the policy maker has a discrete uniform prior $\pi_\mu$ on $\mathcal{M}$, meaning that $\pi_{\mu}(\mu = \Bar{\mu}) = \pi_{\mu}(\mu = -\Bar{\mu}) = 1/2$.

Just as we did in Section \ref{subsection:MMR}, we consider the regret loss $L(a,\theta) \equiv \mu^\star [\mathbf{1}\{\mu^\star \geq 0\}-a]$ and the risk function
\[R(d,\theta) \equiv \mathbb{E}_\theta[L(d,\theta)].\]
However, we are now interested in the average (or Bayesian) risk of a decision rule defined as
\[ r(d,\pi) \equiv E_{\pi} [ R(d,\theta) ]. \]
Let $\mathcal{D}^*$ be the set of all decision rules. The minimax problem of interest is thus
\begin{equation} \label{eq:robust_bayes_problem}
    \inf_{d \in \mathcal{D}^*} \sup_{\pi \in \Gamma} r(d,\pi).
\end{equation}
We follow the literature and refer to any decision rule that solves this problem as either ex-ante $\Gamma$-minimax or ex-ante Robust Bayes. 

\cite{aradillas2024robust} showed that, under some conditions, the problem in Equation \eqref{eq:robust_bayes_problem} for the two-point priors on $\mu$ described before has infinitely many solutions. One such solution takes the form
\[d^\star(\hat{\mu}) = \begin{cases}
    0, & \hat{\mu} < -\frac{\sigma^2 \rho^\star}{\Bar{\mu}}\\
    \frac{\Bar{\mu}\hat{\mu} + \sigma^2 \rho^\star}{2\sigma^2 \rho^\star}, & -\frac{\sigma^2 \rho^\star}{\Bar{\mu}} \leq \hat{\mu} \leq \frac{\sigma^2 \rho^\star}{\Bar{\mu}}\\
    1, & \hat{\mu} > \frac{\sigma^2 \rho^\star}{\Bar{\mu}}
\end{cases},\]
where $\rho^\star$ uniquely solves
\[\int_0^1 \Phi\left(\frac{2\rho^\star x - \rho^\star - (\Bar{\mu}/\sigma)^2}{\Bar{\mu}/\sigma}\right)dx = \frac{-\Bar{\mu} + k}{2k}.\]

We compare this $\Gamma$-minimax optimal rule with the $\epsilon$-approximation obtained via the Hedge algorithm. We again consider the class $\mathcal{D}$ of decision rules of the form
\[d_i = \mathbf{1}\{\hat{\mu} \geq c_i\},\]
where $c_i \in \mathbb{R}$. We again start with an equally spaced grid of 501 points over $[-k,k]$.

In order to apply the Hedge algorithm we extend the Bayes risk $r(d,\pi)$ to any element $p \in \Delta^{I-1}$ by defining
\[ r(p,\pi) \equiv \sum_{i=1}^{I} p_i r(d_{i},\pi) = \sum_{i=1}^I p_i \mathbb{E}_{\pi}\left[ R(d_i,\mu,\mu^\star)\right] = \mathbb{E}_{\pi}\left[\sum_{i=1}^I p_i R(d_i,\mu,\mu^\star)\right]. \] 

We note that Assumption \ref{asn:A1_Risk} is satisfied with the same $M$ as in Subsection \ref{subsection:MMR}. In order to verify Assumption \ref{asn:A2_oracle}, we note that the results in \cite{giacomini2021robust} show that
\begin{equation} \label{eqn:worst-case-Bayes-risk}
\sup_{\pi \in \Gamma} \mathbb{E}_{\pi}\left[\sum_{i=1}^I p_i R(d_i,\mu,\mu^\star)\right] = \mathbb{E}_{\pi_\mu}\left[\Bar{\Lambda}(\mu,p_1,...,p_I)\right],
\end{equation}
where 
\begin{equation} \label{eqn:Bar_lambda_GK}
\Bar{\Lambda}(\mu,p_1,...,p_I) \equiv \sup_{\mu^\star \in [\mu-k,\mu+k]} \sum_{i=1}^I p_i R(d_i,\mu,\mu^\star).
\end{equation}
This relation immediately gives the prior $\pi \in \Gamma$ associated to the worst-case Bayes risk of any vector $p \in \Delta^{I-1}$. In particular, the prior $\pi^p \in \Gamma$ that achieves the worst-case Bayes risk in \eqref{eqn:worst-case-Bayes-risk} sets the marginal prior over $\mu$ to be $\pi_{\mu}$, and the conditional prior of $\mu^*|\mu$ to be a point mass concentrated in the argument that maximizes \eqref{eqn:Bar_lambda_GK}. Thus, the subgradient $g$ used in the mirror descent update is 
\begin{equation} \label{eqn:gradient_robust_bayes}
g \equiv \left( \mathbb{E}_{\pi^p}[R(d_1,\mu,\mu^*)], \ldots,\mathbb{E}_{\pi^p}[R(d_I,\mu,\mu^*)] \right). 
\end{equation}
When $\pi_{\mu}$ has a discrete uniform prior supported on the set $\mathcal{M} = \{-\Bar{\mu},\Bar{\mu}\}$, the $i$-th coordinate of $g$ is
\[g_t^i = (1/2) \cdot R(d_i,\Bar{\mu},\Bar{\mu}_t^\star) + (1/2) \cdot R(d_i,-\Bar{\mu}, (-\Bar{\mu})_t^\star),\]
where $\Bar{\mu}_t^\star$ and $(-\Bar{\mu})_t^\star$ are the corresponding values of $\mu^\star$ for $\mu=\Bar{\mu}$ and $\mu=-\Bar{\mu}$, that solve \eqref{eqn:Bar_lambda_GK}.
We can show that the solutions of $\mu^*$ (as a function of $\mu$) are given by
\[\mu^\star = \begin{cases}
    \mu + k, & \frac{\mu + k}{2k} \geq \sum_{i=1}^I p_i \Phi\left(-\frac{c_i-\mu}{\sigma}\right)\\
    \mu - k, & \frac{\mu + k}{2k} < \sum_{i=1}^I p_i \Phi\left(-\frac{c_i-\mu}{\sigma}\right)
\end{cases}.\]

We consider the case in which $\sigma = 1$, $k = 3$, and $\Bar{\mu} = 0.5$. We set $\epsilon = 0.15$. The number of epochs in Theorem \ref{thm:MWU} is again
\[ T = \lceil 2 M^2 \ln (I)/ \epsilon^2 \rceil = 9{,}974. \]
The algorithm runs for $T = 9{,}974$ iterations and finishes in about $0.070$ seconds (MATLAB R2023a on an Intel Core i9-13900H).

Figure \ref{fig:epsilon_minimax2} shows the true solution versus its $\epsilon$-approximate solution. Qualitatively, the two are very close. The minimax values are close as well, with the $\epsilon$-approximation having a minimax value of $1.458381$ and the true solution having a minimax value slightly lower at $(k^2-\bar\mu^2)/(2k)=35/24\simeq1.458333333$.  Note that here, the term referred to as $\rho^\star$-adjusted is
\begin{equation} \label{eq:rho-star}\rho^\star\text{-adjusted} = \frac{\sigma^2 \rho^\star}{\Bar{\mu}}.
\end{equation}

\begin{figure}[htbp]
\centering
\includegraphics[width=.485\linewidth]{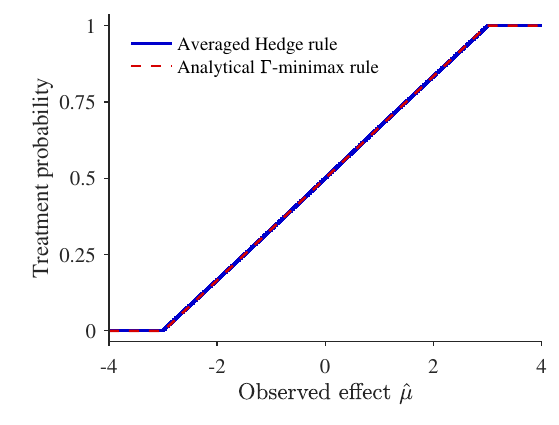}
\caption{Treatment probabilities for the averaged Hedge rule (solid blue) and
the analytical $\Gamma$-minimax rule (dashed red). The two curves nearly
coincide and are difficult to distinguish. Here,
$\rho^\star\text{-adjusted}=\sigma^2\rho^\star/\bar{\mu}
\simeq 2.987884154$ is the value of $\hat{\mu}$ at which the analytical
rule reaches treatment probability one.}
\label{fig:epsilon_minimax2}
\end{figure}

\begin{remark}[\emph{Stochastic Mirror Descent}] \label{remark:SMD} In the ex-ante Robust Bayes problem above, we focused on a uniform ``two-point'' prior for $\mu$ supported on the set $\mathcal{M} = \{-\Bar{\mu},\Bar{\mu}\}$. This assumption simplified considerably the evaluation of the subgradient $g$ in \eqref{eqn:gradient_robust_bayes}. For more general priors over the real-valued parameter $\mu$, the exact evaluation of the gradient is more challenging. However, we note that in this problem it is still possible to implement a \emph{stochastic} version of mirror subgradient descent, by using an unbiased estimator of $g$; see Chapter 6 and Theorem 6.1 in \cite{bubeck2015convex}. See also \cite{fernandez2025approximate} and \cite{montielolea2026bounded}. In the context of the example, one could obtain such an unbiased estimator by using a single draw from $\pi^p$. More concretely, let $\tilde{\mu}$ be a draw from the prior $\pi_{\mu}$. Let $\tilde{\mu}^*$ be the argument that maximizes \eqref{eqn:Bar_lambda_GK} when $\mu=\tilde{\mu}$. Then
\[ \tilde{g} \equiv \left( R(d_1,\tilde{\mu},\tilde{\mu}^*), \ldots, R(d_{I},\tilde{\mu},\tilde{\mu}^*)  \right) \]
is an unbiased estimator of $g$. Under Assumption \ref{asn:A1_Risk}, this unbiased estimator is bounded for every realization of $\tilde{\mu}$. Theorem 6.1 in \cite{bubeck2015convex} allows us to show that the decision rule obtained via stochastic mirror descent is \emph{on average} an $\epsilon$-approximate solution. More precisely, let $\tilde{p}^\epsilon$ be the decision rule obtained for the Robust Bayes problem from Algorithm \ref{alg:mirror_descent} when $\tilde{g}$ is used instead of $g$. Then
\[ \bar{v} \equiv \inf_{p \in \Delta^{I-1}} \sup_{\pi \in \Gamma} r(p,\pi) \leq \tilde{\mathbb{E}}\left[\sup_{\pi \in \Gamma} r(\tilde{p}^\epsilon,\pi) \right] \leq \tilde{v} + \epsilon,   \]
where the average, $\tilde{\mathbb{E}}$, is taken over different potential runs of stochastic mirror descent.
\end{remark}

\section{Application} \label{sec:application}

\cite{lee2021poverty} conducted a randomized controlled trial in Bangladesh to estimate the effects of encouraging rural households to receive money transfers from migrant family members. They specifically conducted an encouragement design where poor rural households with family members who had migrated to a larger urban destination received a 30--45 minute training about how to register and use the mobile banking service ``bKash'' to send instant remittances back home.

The experiment was conducted in the Gaibandha district, one of Bangladesh's poorest regions. It focused on households that had migrant workers in the Dhaka district, the administrative unit in which the capital of Bangladesh is located. \cite{lee2021poverty} measure several outcomes of both receiving households and sender migrants; see their Figures 3 and 4. To give a concrete example of the measured outcomes, one question of interest is whether families that adopt the mobile banking technology are more (or less) likely to declare that the \emph{monga}---the seasonal period of hunger in September through November---was not a problem for their household. Table 9, Column 7, p. 60 in \cite{lee2021poverty} presents results for this specific variable showing that households that used a bKash account in the treatment group are 9.2 percentage points more likely to declare that \emph{monga} was not a problem.\footnote{The effect of 9.2 percentage points is, technically speaking, the effect of active bKash use for \emph{compliers}. To simplify the analysis, we assume that all the units are compliers. The number is used as an illustrative magnitude.} The standard error of the estimator is 4.5 percentage points.  

Is the corridor selected by \cite{lee2021poverty} a good choice for a researcher who is concerned about external validity?\footnote{Following \cite{gechter2023site}, we name the corridors using a destination-origin format; for example, the migration corridor studied in \cite{lee2021poverty} is ``Dhaka-Gaibandha''.} Two recent papers provided answers to this question. \cite{gechter2023site} use an elegant decision-theoretic framework to argue that the \emph{Dhaka-Noakhali} corridor would have been a better choice from the perspective of maximizing average welfare. \cite{olea2024externally} use the framework of \cite{gechter2023site} to argue that the \emph{Dhaka-Pabna} corridor would have been a better choice from the minimax (welfare) regret criterion perspective (restricting the policy maker to consider only nonrandomized selection of corridors). The \emph{Dhaka-Pabna} corridor is also recommended by the \emph{synthetic purposive sampling} approach in \cite{egamidesigning}. One important comment is that the \emph{Dhaka-Pabna} corridor is the most representative in terms of covariates, in the sense that it minimizes the average distance (measured using the Euclidean distance between covariates) to the $41$ migration corridors analyzed in \cite{gechter2023site}. 

In our application, we consider a situation where a policy maker is considering the three sites mentioned above to run an experiment: Dhaka-Gaibandha (the original site in \cite{lee2021poverty}), Dhaka-Noakhali (the site suggested by \cite{gechter2023site}) and Dhaka-Pabna (the site suggested in \cite{olea2024externally}). Each of these sites (migration corridors) has site characteristics $X_s \in \mathbb{R}^{d}$, with $d=13$.\footnote{The covariates include mean household income, mean household size, migrant density, mean remittances. See Figure 2 in \cite{olea2024externally}.} We index these three sites by $1,2,3$ respectively and refer to the set $\mathcal{S}_E \equiv \{1,2,3\}$ as the set of experimental sites. Once we exclude these three sites, we have 38 migration corridors. We use the distance between the covariates of each of these sites and Dhaka-Pabna to order them in increasing order and index them with integers $4$ to 41. Figure \ref{fig:distances_experimental_policy} presents the distances. The figure shows that for most of the sites the corridor Dhaka-Pabna is the ``closest'' in terms of the Euclidean distance between covariates.  

\begin{figure}[h!]
    \centering
    \includegraphics[scale=0.7]{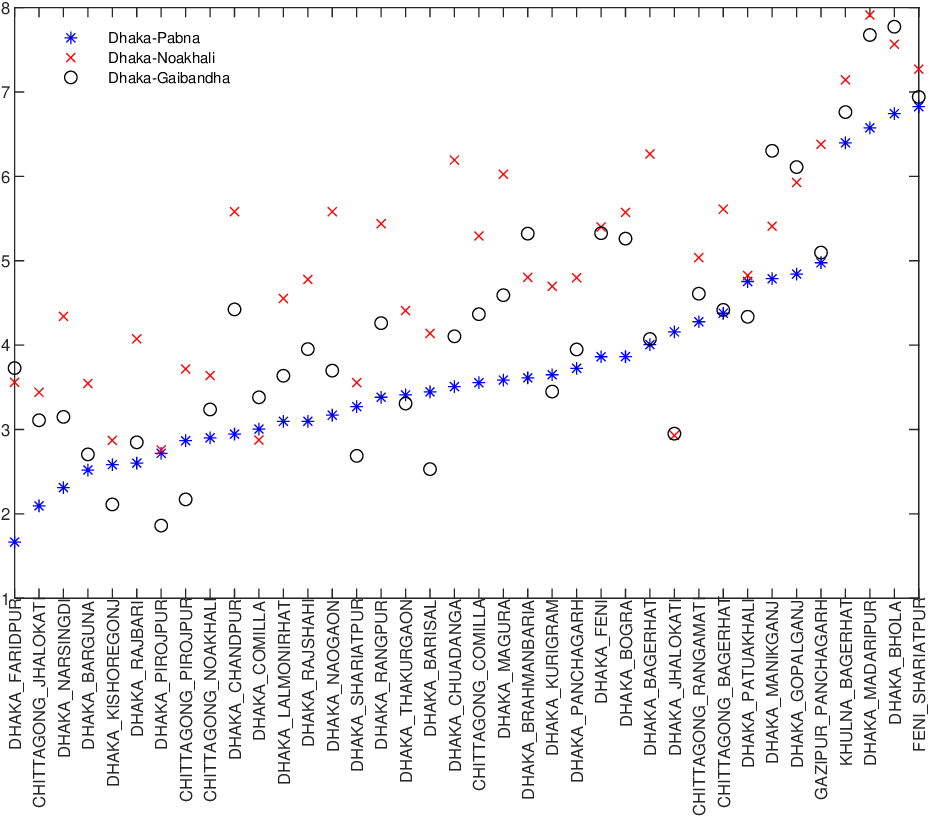}
    \caption{Distances from each of the experimental sites to each of the policy-relevant sites.}  
    \label{fig:distances_experimental_policy}
\end{figure}

We assume that the sites $\mathcal{S}_{P} \equiv \{ 4, \ldots, 41 \}$ on the $x$-axis of Figure \ref{fig:distances_experimental_policy} are the policy-relevant sites. This means that the policy maker is interested in deciding whether the training program discussed in \cite{lee2021poverty} should be rolled out in these sites. We assume that the outcome variable of interest for the policy maker is the likelihood that the households declare that the \emph{monga} was not a problem. 

\emph{Treatment Effect Heterogeneity:} 
Treatment effect heterogeneity is allowed, but only via the observable site characteristics. The effects of the policy in each site, denoted by $\tau_{s}$, are restricted to be a Lipschitz function (with respect to a Euclidean norm $||\cdot||$) with known constant $C$; that is, $\tau_{s} = \tau(X_s)$, where
\[|\tau(x)-\tau(x')| \leq C||x-x'||, \hspace{1em} \forall x,x' \in \mathbb{R}^{13}.\]
One first issue that we need to address in order to conduct our exercise is the value of $C$ that will be used in our application. We do this by using the available point estimates of the treatment effect of the program in \cite{lee2021poverty}. Let $x_{DG}$ denote the covariates of the corridor Dhaka-Gaibandha. Assume that  we entertained the possibility that the true effect, $\tau(x_{DG})$, coincides with the estimated effect 9.2. We consider a ``low $C$'' regime. 

Suppose that we want to consider a value of $C$ that imposes that if 9.2 were the true effect, then even the corridor that is the most different (in terms of covariates) to Dhaka-Gaibandha the effect of the program must be nonnegative. Dhaka-Bhola is the most different site and $\| x_{\textrm{DG}}- x_{\textrm{DB}} \| = 7.7736$. Since the Lipschitz restriction imposes that 
\[  \tau(x_{\textrm{DG}}) - C \| x_{\textrm{DG}}- x_{\textrm{DB}} \|   \leq \tau(x_{\textrm{DB}}),\]
we could pick $C$ as 
\[ C = 9.2 / 7.7736 \approx 1.1835. \]

For comparison, in Appendix \ref{app:lipshitz_sensitivity} we also consider a larger value of $C$, calibrated
using Dhaka--Pirojpur, the closest other corridor to Dhaka--Gaibandha.
Their covariate distance is $1.8614$, giving
$C=9.2/1.8614\approx4.9425$. If the true effect in Dhaka--Gaibandha
were $9.2$, this value would imply a nonnegative effect in
Dhaka--Pirojpur while allowing negative effects in more distant
corridors.

\emph{Treatment Rules:} The policy maker makes two choices. First, the policy maker must pick one site on which to experiment. Second, the policy maker must decide how to make treatment choices in all the sites of interest given the available data. We assume that if the policy maker decides to experiment on site $s$, the available data becomes $\widehat{\tau}_s$, with
\begin{equation} \label{eq:site-selection-data}\widehat{\tau}_{s} \sim \mathcal{N}(\tau_s,\sigma^2_s)
\end{equation}
and, as in \cite{gechter2023site}, we assume $\sigma^2_s$ is known. In order to conduct our exercise, we assume that $\sigma_s$ is the same for all experimental sites, and that it matches the standard error of the estimated effect of the program in the Dhaka-Gaibandha site. That is $\sigma_{s}=4.5$ for all $s \in \mathcal{S}_{E}$.  

The treatment rule is a mapping $T:\mathbb{R} \rightarrow [0,1]^{\#\mathcal{S}_{P}}$. For $s \in \mathcal{S}_{E}$ we further denote by $T_s$ the specific policy choice for site $s$. We refer to a tuple $(s,T)$ as a policy, and we use $d$ to denote it. 
We consider three nonrandomized policies
\[\mathcal{D} \equiv \left\{d_1, d_2, d_3 \right \}. \]
Under policy $d_s$, the policy maker experiments on site $s \in \mathcal{S}_{E}$ and its recommendation for any policy relevant site is $\mathbf{1}\{\widehat{\tau}_s \geq 0\}$. That is, the policy maker makes the same policy recommendation for every policy-relevant site depending on the sign of $\widehat{\tau}_{s}$.\footnote{The results in \cite{olea2024externally} suggest that this type of policy is likely to be suboptimal. The policy maker could improve its welfare by allowing the treatment choice to be randomly selected, depending on the distance between the policy-relevant site and the experimental site.} We focus on this special form of policy rule because we think it captures the policy recommendations that are given based on randomized controlled trials.

We consider the following regret function for the policy $d_s$, 
\begin{equation} \label{eq:site-selection-regret}
\begin{aligned}
R(d_s,\tau) &\equiv \frac{1}{\#\mathcal{S}_{P}} \sum_{s' \in \mathcal{S}_{P}} \left(\tau(X_{s'})(\mathbf{1}\{\tau(X_{s'}) \geq 0\} - \mathbb{E}_{\tau(X_s)}\left[\mathbf{1}\{\widehat{\tau}_s \geq 0\}\right]) \right).
\end{aligned}
\end{equation}
This expression can be simplified to
\begin{equation}\label{eq:site-selection-simplified-regret}
\begin{aligned}
R(d_s,\tau) &\equiv \frac{1}{\#\mathcal{S}_{P}} \sum_{s' \in \mathcal{S}_{P}} \left(\tau(X_{s'})\left(\mathbf{1}\{\tau(X_{s'}) \geq 0\} - \Phi\left(\frac{\tau(X_s)}{\sigma_s}\right)\right) \right).
\end{aligned}
\end{equation}

The minimax (regret) problem that we are interested in solving is
\begin{equation} \label{eq:site-selection-minimax-problem}\inf_{p \in \Delta^{2}} \sup_{\tau \in \text{Lip}_{C}(\mathbb{R}^{13})} \sum_{s=1}^{3} p_{s} R(d_s,\tau),
\end{equation}
where $\text{Lip}_{C}(\mathbb{R}^{13})$ refers to the space of all Lipschitz functions $f:\mathbb{R}^{13} \rightarrow \mathbb{R}$ with constant $C$. For a given $p\in\Delta^2$, nature's objective depends on $\tau$
only through its values at the experimental and policy-relevant sites.
Writing $\mathcal S=\mathcal S_E\cup\mathcal S_P$ and $u_j \equiv \tau(X_j)$,
we therefore maximize the weighted regret in
\eqref{eq:site-selection-minimax-problem} over $(u_j)_{j\in\mathcal S}$
subject to
\[
 |u_j-u_k|\leq C\|X_j-X_k\|,\qquad j,k\in\mathcal S.
\]
We compute numerical best responses using MATLAB's \texttt{fmincon},
initialized at the zero vector, to minimize the negative weighted regret. The objective function that defines nature's best responses is not concave, and
\texttt{fmincon} cannot provide a certificate for an exact or approximate solution. This means that our algorithm can only provide a provably $\epsilon$-minimax decision rule under the assumption that \texttt{fmincon} finds an exact solution. The decision rule could also be $\epsilon+\delta$-minimax if we could provide a $\delta$ bound on the error of \texttt{fmincon}, as explained in Remark~\ref{remark:relax_A3}. We also calculate an upper bound of $M$ by maximizing
 the objective in Equation \eqref{eq:site-selection-minimax-problem}
 over both $p\in\Delta^2$ and
 $\tau\in\text{Lip}_{C}(\mathbb{R}^{13})$ using \texttt{fmincon}. In what follows, we keep working under the assumption that \texttt{fmincon} verifies Assumption \ref{asn:A2_oracle} and obtains $M$. Thus, we refer to the output of the algorithm as $\epsilon$-minimax for the rest of this section.

\subsection{Results} 

We report results for the case in which $C=1.1834$.\footnote{
Using the larger value $C=4.9425$ gives qualitatively
similar site-selection results in the one- and five-policy-site
cases; see Appendix \ref{app:lipshitz_sensitivity}.} We consider four different scenarios that vary in terms of the number of sites that are policy relevant. The scenarios we consider have either 1, 5, 15, or 38 policy-relevant sites. In each of these cases, we choose to include the sites that are closest to Dhaka-Pabna. For example, when we include only one policy-relevant site we include Dhaka-Faridpur. We do this because, in light of the results in \cite{olea2024externally}, the best nonrandomized choice of experimental site is Dhaka-Pabna. And we would like to use this application to understand how the probability of selecting this site changes as we include sites that perhaps are closer to some of the other experimental sites under consideration. 

Figure \ref{fig:minimax_probabilities} presents the $\epsilon$-minimax selection of sites obtained via the Hedge algorithm. Note first that when there is only one policy-relevant site (and this site is closest to Dhaka-Pabna) the probability of choosing Dhaka-Pabna is close to 1. This is measured by the height of the first yellow bar in Figure \ref{fig:minimax_probabilities}. We think this is an interesting result as it shows that even if randomization is allowed, it is possible that choosing the site that is most representative for the policy-relevant site is still  approximately minimax regret optimal. 

The results with five policy-relevant sites are also worth discussing. By construction, the five policy-relevant sites that we considered are those that are closest to Dhaka-Pabna. According to Figure \ref{fig:distances_experimental_policy}, Dhaka-Pabna is the nearest neighbor for all of them, with the exception of Dhaka-Kishoregonj. For the latter site, the nearest neighbor is Dhaka-Gaibandha. Figure \ref{fig:minimax_probabilities} shows that, with 5 sites, the $\epsilon$-minimax selection of experimental sites places probability slightly above $.2$ on Dhaka-Gaibandha (corresponding to the height of the second blue bar) and probability close to $.7$ on Dhaka-Pabna (corresponding to the height of the second yellow bar).    

\begin{figure}[h!]
    \centering
    \includegraphics[scale=0.60]{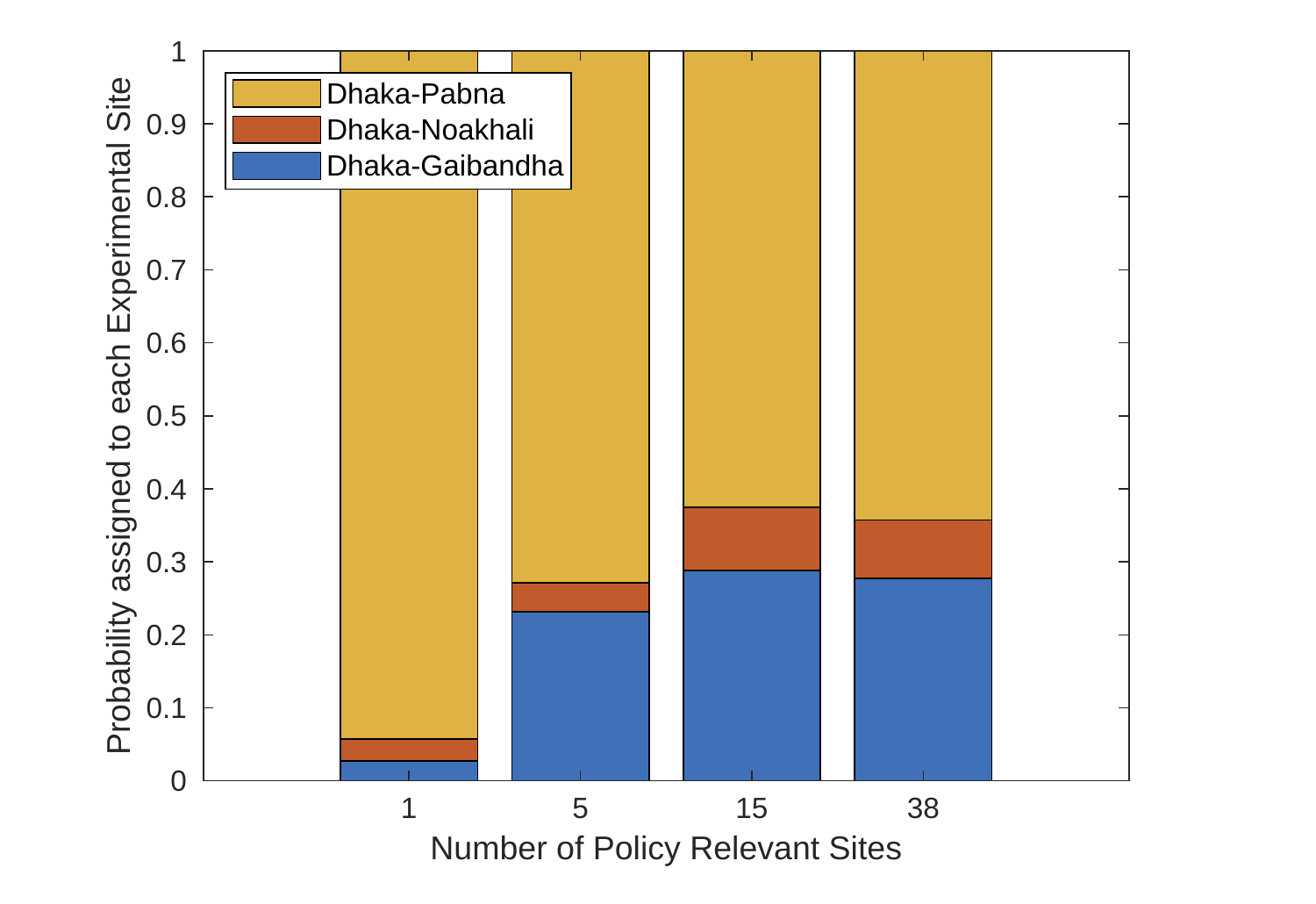}
    \caption{$\epsilon$-Minimax decision rule for the Site Selection Problem via the Hedge algorithm. The graph is generated using $C=1.1834$, $\sigma = 4.5$, and $\epsilon=.1$. }  
    \label{fig:minimax_probabilities}
\end{figure}

We finally discuss the cases in which there are 15 or 38 policy-relevant sites. The corresponding $\epsilon$-minimax solutions are very similar, though the computation times and numbers of iterations are not; see Table \ref{table:Runtime_iterations}. The recommended probability of experimenting on Dhaka-Pabna is close to .6 (height of the last yellow bar). The recommended probability of experimenting on Dhaka-Noakhali is close to .1. Interestingly, the ordering of the probabilities is also consistent with the ordering of the three experimental sites in terms of how frequently they are the nearest neighbor of each of the policy-relevant sites.

\begin{table}[h!] 
    \centering
    \begin{tabular}{|c|c|c|c|}
    \hline
    Number of Sites & Runtime (seconds) & Iterations & Mean Runtime per Iteration\\
    \hline
    1 & $7.14$ & $1,103$ & $0.0065$\\
    5 & $10.73$ & $1,015$ & $0.0106$\\
    15 & $54.50$ & $1,208$ & $0.0451$\\
    38 & $293.09$ & $1,883$ & $0.1556$\\
    \hline
    \end{tabular}
    \caption{Runtime (seconds), Number of Iterations ($C=1.1834$), and Mean Runtime per Iteration.} \label{table:Runtime_iterations}
    \parbox{.96\linewidth}{\footnotesize Computation performed using MATLAB R2023a on an Intel Core i9-13900H runtime includes the numerical best-response calculations
    and Hedge updates; calibration of $M$, final risk evaluation, and file
    exports are excluded, see Appendix \ref{app:lipshitz_sensitivity}.}
\end{table}

We can also reframe our results to start from setting a target runtime for Algorithm \ref{alg:mirror_descent}. We next show how, after picking a desired run time, we can obtain a value of $\epsilon$ that matches this target. 

For the sake of concreteness, suppose that we are willing to spend $\tau = 1,800$ seconds running Algorithm \ref{alg:mirror_descent} for the case of one policy site. For the one-policy-site
case, Table~\ref{table:Runtime_iterations} reports a mean runtime
of $r\equiv .00647$ seconds per iteration, including both the numerical
best-response calculation and the Hedge update. Thus, the total number of iterations that we could afford to run becomes $\lfloor \tau / r \rfloor = 278{,}207$. Since $M^2$ and $I$ are known, we can solve for the value of $\epsilon$ using the formula for the number of iterations, $T$, presented in Algorithm \ref{alg:mirror_descent}:
\[\epsilon(\tau) \equiv \sqrt{\frac{2 \cdot M^2 \cdot \ln(I)}{\lfloor\tau/r\rfloor}}.\]
Using $M^2=5.0183$, $I=3$, and $r=0.00647$ seconds per iteration
gives $\epsilon(1{,}800)\approx0.0063$. Table \ref{table:epsilons_needed} shows various values of $\epsilon$ needed for different runtimes (for different number of policy-relevant sites). These values can be thought of as the maximum amount of precision (lowest $\epsilon$) attainable for each possible number of policy-relevant sites if the policy maker is only willing to spend the amount of compute time in each of the columns of Table \ref{table:epsilons_needed}.

\begin{table}[h!]
    \centering
    \begin{tabular}{|c|c|c|c|c|}
    \hline
    Number of Sites & 30 minutes & 1 hour & 2 hours & 10 hours\\
    \hline
    1 & $0.006$ & $0.004$ & $0.003$ & $0.001$\\
    5 & $0.008$ & $0.005$ & $0.004$ & $0.002$\\
    15 & $0.017$ & $0.012$ & $0.009$ & $0.004$\\
    38 & $0.040$ & $0.029$ & $0.020$ & $0.009$\\
    \hline
    \end{tabular}
    \caption{Planned iteration-error allowance as a function of
    runtime, using the original calibration.}
    \label{table:epsilons_needed}
\end{table}

\section{Conclusion} 

This paper presented an algorithm for obtaining \emph{$\epsilon$-minimax} solutions of statistical decision problems where the statistician is allowed to choose randomly among $I$ decision rules. The notion of an $\epsilon$-minimax decision rule was taken from \cite{Ferguson67} (Chapter 1, Definition 4) and it refers to a decision rule whose worst-case expected loss exceeds the minimax value of the decision problem by at most an additive factor of $\epsilon$.\footnote{We note that the definition given in \cite{Ferguson67} differs from the usage of $\epsilon$-minimax decision rules in other contexts. Most notably, from the work of \cite{manski2016sufficient}, who use the term $\epsilon$-minimax to refer to a decision rule whose worst-case regret is at most $\epsilon$.} 

Once we allow for randomized selection over the $I$ decision rules, the minimax problem admits a convex programming representation over the $(I-1)$-simplex, an observation which has been previously documented in the literature by \cite{chamberlain2000econometric}. Both the objective function and the subgradient of this problem are in general difficult to evaluate, the reason being that the objective function of the convex problem involves solving a nonconvex maximization problem to find the worst-case performance (over the model's parameter space) of a specific randomized selection over the $I$ rules. This type of problem arises commonly in the convex optimization literature; see \cite{bubeck2015convex} and the seminal work of \cite{nemirovski1983problem}. The algorithm herein suggested is a mirror subgradient descent (with negative entropy as a mirror map), initialized with uniform weights and stopped after a finite number of iterations. The early stopping of the algorithm tries to minimize the number of calls to the objective function and its subgradient, but it provably generates an approximate solution with the desired tolerance $\epsilon$.  

The iterative procedure arising from this mirror descent routine described in this paper is known in the computer science literature as the \emph{Hedge Algorithm}, and it is  used in algorithmic game theory as a practical tool to find approximate solutions of two-person zero-sum games. 

The paper applies the suggested algorithm to different minimax problems in the econometrics literature. In some of these problems, the minimax solution is known, and we show numerically that in these examples the $\epsilon$-minimax solution is practically the same as the true minimax solution.  

Finally, we apply the algorithm to the \emph{site selection problem} of \cite{gechter2023site}; namely,  how to select \emph{sites} optimally to maximize the external validity of an experimental policy evaluation. Our algorithm allows the researcher to choose randomly where to experiment, while adjusting optimally for the available baseline covariate information.  

We think there are several interesting areas for future work, both from an applied and from a more theoretical perspective. From a purely applied angle, our algorithm could be useful in approximately solving certain minimax problems, such as the one described in the recent work of \cite{armstrong2024adapting}. 

From a more theoretical perspective, it would be interesting to further explore the differences between $\epsilon$-minimax strategies and the notion of a local min-max point in \cite{daskalakis2021complexity}. There are very interesting results about the relation between this notion and the stationary points of subgradient ascent-descent dynamics. But it would be great to understand, theoretically and empirically, what the potential benefits are of searching for these types of points as opposed to $\epsilon$-minimax strategies.  
\newpage

\bibliographystyle{ecta}
\bibliography{references.bib}

\appendix

\section{Proofs of Main Results}

\subsection{Proof of Lemma \ref{lemma:Lemma_convex}} \label{subsection:Proof_Lemma_convex}

Take $p,p' \in \Delta^{I-1}$. Note that 
\begin{eqnarray*}
f(\alpha p + (1-\alpha)p') &=& \sup_{\theta \in \Theta} R(\alpha p + (1-\alpha)p',\theta)\\
&=& \sup_{\theta \in \Theta}  \sum_{i=1}^{I} (\alpha p_i + (1-\alpha)p_{i}') R(d_i,\theta)  \\
&=& \sup_{\theta \in \Theta} \left( \alpha \sum_{i=1}^{I} \left[p_i R(d_i,\theta) + (1-\alpha) p_{i}' R(d_i,\theta) \right] \right) \\
& \leq& \alpha  \sup_{\theta \in \Theta}  \sum_{i=1}^{I} p_i R(d_i,\theta) +  (1-\alpha)\sup_{\theta \in \Theta} \sum_{i=1}^{I} p_{i}' R(d_i,\theta) \\
&=& \alpha f(p)  + (1-\alpha) f(p'). 
\end{eqnarray*}
Thus, the function $f(\cdot)$ is convex. 

Now we establish the Lipschitz continuity of $f(\cdot)$. Take any $p,p' \in \Delta^{I-1}$. Note that
\begin{align*}
    |f(p) - f(p')| & = | \sup_{\theta \in \Theta} \sum_{i=1}^{I} p_{i}  R(d_i, \theta) - \sup_{\theta \in \Theta} \sum_{i=1}^{I} p'_{i}  R(d_i, \theta) |\\
    & \leq  \sup_{\theta \in \Theta} | \sum_{i=1}^{I} p_{i}  R(d_i, \theta)-\sum_{i=1}^{I} p'_{i}  R(d_i, \theta)| \\
    & \leq\sup_{\theta \in \Theta} \sum_{i=1}^{I} |(p_{i} - p'_{i}) R(d_i, \theta) |\\
    & \leq M \sum_{i=1}^{I} |p_{i} - p'_{i} |  \\
    & = M \|p - p' \|_{1}, 
\end{align*}
where the final inequality applies Assumption~\ref{asn:A1_Risk}. This shows that $f(\cdot)$ is Lipschitz continuous with constant at most $M$. 

We now show that $g_0$ is a subgradient of $f$ at $p_0$. That is, that for any $p \in \Delta^{I-1}$:
\[ f(p) \geq f(p_0) + g_0^{\top} (p-p_0). \]
Let $p$ be an arbitrary point in $\Delta^{I-1}$. By definition
\begin{eqnarray*}
f(p_0) &\equiv& \sup_{\theta \in \Theta} \sum_{i=1}^{I} p_{0,i} R(d_i,\theta) \\ 
&=& \sum_{i=1}^{I} p_{0,i} R(d_i,\theta_0) \\
&=& \sum_{i=1}^{I} (p_{0,i}-p_i) R(d_i,\theta_0)  + \sum_{i=1}^{I} p_i R(d_i,\theta_0) \\
&=& g_0^{\top} (p_0-p) + \sum_{i=1}^{I} p_i R(d_i,\theta_0) \\
&\leq& g_0^{\top} (p_0-p) + f(p). 
\end{eqnarray*}
Thus, $g_0$ is a subgradient of $f(\cdot)$ at $p_0$.

\subsection{Proof of Theorem \ref{thm:MWU} via \cite{bubeck2015convex}} \label{subsection:Proof_Thm_1_bubeck}
We apply Theorem 4.2 from \cite{bubeck2015convex}. In order to do so, we note that: 
\begin{enumerate}
    \item The mirror map used in our mirror-descent routine is the negative entropy $\Phi(p) = \sum_{i=1}^{I} p_i \log p_i$, defined for any $p$ in the relative interior of $\Delta^{I-1}$, which abusing notation we denote as $\textrm{int}(\Delta^{I-1})$. By Pinsker's inequality,  we have that for any $p,q \in \textrm{int}(\Delta^{I-1})$: 
    \[ \sum_{i=1}^{I} p_i \log p_i - \sum_{i=1}^{I} p_i \log q_i \geq \frac{1}{2} \left( \sum_{i=1}^{I} | p_i - q_i | \right)^2. \]
    This can be written as 
    \begin{eqnarray*}
    \Phi(q) - \Phi(p) &\leq& \sum_{i=1}^{I} (q_i-p_i) \log q_i  - \frac{1}{2} \| p-q \|^2_1 \\
    &=& \sum_{i=1}^{I} (q_i-p_i) \left( 1+ \log q_i \right)    - \frac{1}{2} \| p-q \|^2_1.
    \end{eqnarray*}
    Since the gradient of $\Phi(\cdot)$ with respect to $q$ is $(1+ \ln q_1, \ldots, 1+ \ln q_{I})^{\top}$, this means that the negative entropy is $\rho$-strongly convex over $\textrm{int}(\Delta^{I-1})$ (with respect to $\|\cdot \|_1$) with parameter $\rho = 1$.
    \item We calculate the ``radius'' of $\textrm{int}(\Delta^{I-1})$ defined as $R^2 \equiv \sup_{p \in \textrm{int}\left( \Delta^{I-1} \right)} \Phi(p) - \Phi(p_0)$. We show that $R^2 = \ln(I)$. To do this, note that
    \[ \sup_{p \in \textrm{int}\left( \Delta^{I-1} \right) } \Phi(p) = 0. \] 
    The supremum is attained by any sequence of distributions that converges to a discrete distribution that places all of its mass in one of point. Then, since $p_0 = (1/I, ..., 1/I)$ in our setting, we have that $\Phi(p_0) = -\ln(I)$. This implies that $R^2 = \ln(I)$. 
    \item Lemma~\ref{lemma:Lemma_convex} has implied that the objective function $f$ is convex and $L$-Lipschitz continuous w.r.t. $\|\cdot\|_1$, with parameter $L = M$. Let $p^*$ be a minimizer of $f$ over $\Delta^{I-1}$.
\end{enumerate}
In light of 1-2-3, the conditions of Theorem 4.2 in \cite{bubeck2015convex} are verified. The theorem then implies that mirror descent with step size $\eta$
satisfies 
\begin{equation*}
    f \left( \frac{1}{T} \sum_{t = 1}^{T} p_t \right) - f(p^{\star}) \leq 
    \frac{1}{T}\sum_{t=1}^T f(p_t) - f(p^{\star}) \leq
    \frac{\ln(I)}{T \eta} + \frac{\eta M^2 }{2}.
\end{equation*}
Notice that $f(p^{\star}) = \bar{v}$.
For a given error $\epsilon$, we choose $\eta = \epsilon/M^2$ and $T \geq 2 M^2  \ln(I)/\epsilon^2$ such that the right hand side is smaller than $\epsilon$:
\begin{eqnarray*}
    \frac{\ln(I)}{T \eta} + \frac{\eta M^2 }{2} = \frac{\ln(I)M^2}{T \epsilon} + \frac{\epsilon}{2} \leq \epsilon.
\end{eqnarray*}
Therefore, we conclude with
\begin{equation*}
    f(p^\epsilon_i) - \bar{v} \leq \bar{v}^{\epsilon} - \bar{v} \leq \epsilon,
\end{equation*}
where  $\bar{v}^{\epsilon} \equiv \frac{1}{T} \sum_{t=1}^{T} \left( \sum_{i=1}^{I} p_{i,t} R(d_i,\theta_t) \right)$, and $\theta_{t}$ corresponds to ``nature's best response'' in the $t$-th iteration of the mirror descent routine in Algorithm \ref{alg:mirror_descent}.

\newpage

\part*{Online Appendix}
\global\long\def\thepage{OA-\arabic{page}}\setcounter{page}{1}

\section{Additional Technical Results}

\subsection{Alternative Approximations to the Minimax Value}
\label{remark:alternative-value}
By definition of $\epsilon$-minimax decision rule, 
\[ \bar{v} \leq \sup_{\theta \in \Theta} R(p^\epsilon,\theta) \leq \bar{v}+\epsilon. \]
Thus, the worst-case risk of $p^\epsilon$ provides an approximation to the minimax value $\bar{v}$. We note that there is an alternative approximation that can be obtained directly from the output of Algorithm \ref{alg:mirror_descent}:
\[ \bar{v}^{\epsilon} = \frac{1}{T} \sum_{t=1}^{T} \left( \sum_{i=1}^{I} p_{i,t} R(d_i,\theta_t) \right),  \]
where $\theta_{t}$ corresponds to ``nature's best response'' in the $t$-th iteration of the mirror descent routine in Algorithm \ref{alg:mirror_descent}. See Appendix \ref{subsection:Proof_Thm_1}. That such an approximation to the minimax value is valid is not obvious since
\[ \sup_{\theta \in \Theta} R(p^\epsilon,\theta) = \sup_{\theta \in \Theta} \sum_{i=1}^{I} p^\epsilon_{i} R(d_i,\theta) = \sup_{\theta \in \Theta} \sum_{i=1}^{I} \left( \frac{1}{T}\sum_{t=1}^{T} p^\epsilon_{i,t} \right) R(d_i,\theta) \leq \bar{v}^{\epsilon}.  \]

\subsection{
Proof of Theorem \ref{thm:MWU} via \cite{arora2012multiplicative}} \label{subsection:Proof_Thm_1}

We extend Theorem 2.1 and Theorem 2.3 in \cite{arora2012multiplicative}. For the sake of exposition, we divide our proof into three steps. 

    {\scshape{STEP 1:}} Fix the step-size $\eta$. First, we show that after $T$ rounds the algorithm guarantees that, for all decision rules $d_i \in \{d_1,...,d_I\}$, we have obtained average payoffs  bounded above by the average payoff of any decision rule $d_i$ plus an error term. More precisely:
\begin{equation} 
\label{eqn:aux_step_1}
\frac{1}{T}\sum_{t=1}^T \left(\sum_{i=1}^I p_{i,t} R(d_i,\theta_t)\right) \leq \frac{1}{T}\sum_{t=1}^T R(d_i,\theta_t) + \frac{M^2 \eta}{2} + \frac{\ln(I)}{T \eta}.
\end{equation}
To show this, we use a similar argument to Theorem 2.1 in \cite{arora2012multiplicative}. Let $M>0$ be the bound on the risk function in Assumption \ref{asn:A1_Risk}. We define the normalized subgradient, $m_t \equiv g_t/M$. 
Then, recall the definition of $\phi_t$, we have that 
 \begin{align*}
    \phi_{t+1} &= \sum_{i=1}^I w_{i,t} \\
    &= \sum_{i=1}^I w_{i,t-1}\exp \left( - \eta g_{i,t}\right)\\
    & = \sum_{i=1}^I w_{i,t-1}\exp \left( - \eta M m_{i,t}\right)\\
    & \leq \sum_{i=1}^I w_{i,t-1}\left(1 - \eta M m_{i,t} +  \frac{\eta^2 M^2 m_{i,t}^2}{2} \right)\\
    &= \phi_t - \phi_t \eta M \sum_{i=1}^I p_{i,t} m_{i,t} + \phi_t \frac{\eta^2 M^2}{2} \sum_{i=1}^I p_{i,t} m_{i,t}^2  \\
    &= \phi_t \left(1 -\eta M \sum_{i=1}^I p_{i,t} m_{i,t} + \frac{\eta^2 M^2}{2} \sum_{i=1}^I p_{i,t} m_{i,t}^2\right)\\
    &\leq \phi_t \exp\left( -\eta M \sum_{i=1}^I p_{i,t} m_{i,t} + \frac{\eta^2 M^2}{2} \sum_{i=1}^I p_{i,t} m_{i,t}^2\right).
\end{align*}
The first inequality follows from the fact that $\exp( - x ) \leq 1- x + x^2/2$, for $x \geq 0$. The last inequality follows from $1-x \leq e^{-x}$ for all $x \in \mathbb{R}$. By induction after $T$ rounds, and using the fact that $w_{0}$ was initialized to be a vector of ones (i.e., $w_0:= \mathbf{1}$), we have 
\begin{equation}
    \begin{aligned}[b]
    \phi_{T+1} &\leq 
    \phi_1 \exp\left(-\eta M \sum_{t=1}^T \sum_{i=1}^I p_{i,t} m_{i,t} + \frac{\eta^2 M^2}{2} \sum_{t=1}^T \sum_{i=1}^I p_{i,t} m_{i,t}^2\right)\\
    &= I \exp\left(- \eta M \sum_{t=1}^T \sum_{i=1}^I p_{i,t} m_{i,t} + \frac{\eta^2 M^2}{2} \sum_{t=1}^T \sum_{i=1}^I p_{i,t} m_{i,t}^2\right).
    \end{aligned}
    \label{eqn:aux_step_2}
\end{equation}
Also notice that 
\begin{equation} 
    \phi_{T+1} \geq w_{i,T} = \prod_{t=1}^T \exp \left(-\eta g_{i,t}\right),
\label{eqn:aux_step_3}
\end{equation}
After taking logs of both sides in \eqref{eqn:aux_step_2} and \eqref{eqn:aux_step_3}, we have
\[
- \sum_{t=1}^{T} g_{i,t} \leq \frac{\ln(I)}{\eta} - M\sum_{t=1}^T \sum_{i=1}^I p_{i,t} m_{i,t} + \frac{\eta M^2}{2}\sum_{t=1}^T \sum_{i=1}^I p_{i,t} m_{i,t}^2.
\]
Since $m_{i,t} = g_{i,t}/M = R(d_i, \theta_t)/M$, we have 
\begin{align*}
    \frac{1}{T} \sum_{t=1}^T \sum_{i=1}^I p_{i,t} R(d_i, \theta_t) \leq &
    \frac{1}{T} \sum_{t=1}^{T} R(d_i, \theta_t) + \frac{\eta M^2}{2T} \sum_{t=1}^T \sum_{i=1}^I p_{i,t} m_{i,t}^2 + \frac{\ln(I)}{T\eta} \\
    {\leq} & \frac{1}{T} \sum_{t=1}^{T} R(d_i, \theta_t) + \frac{\eta M^2}{2} + \frac{\ln(I)}{T\eta} \\
    & \textrm{(since $m_{i,t}^2 \leq 1$)}\\
    \leq & \frac{1}{T} \sum_{t=1}^{T} R(d_i, \theta_t) + \frac{\epsilon}{2} + \frac{\ln(I)M^2}{T\epsilon} \\
    & \textrm{(since $\eta = \epsilon/M^2$).}
\end{align*}

    {\scshape STEP 2:} Let $p_i^\epsilon \equiv \frac{1}{T}\sum_{t=1}^T p_{i,t}$. We show that after $T$ rounds, we have that:
    \[\Bar{v} \leq \sup_{\theta \in \Theta} \sum_{i=1}^I p_i^\epsilon R(d_i,\theta) \leq  \frac{1}{T}\sum_{t=1}^T \left(\sum_{i=1}^I p_{i,t}R(d_i,\theta_t)\right) \leq \Bar{v} + \frac{\epsilon}{2} + \frac{\ln(I)}{T} \left(\frac{M^2}{\epsilon}\right),\]
    where $\Bar{v}$ is the minimax value of the decision problem. 

    To show this, note that the lower bound holds by definition. For the upper bound, note:
    \begin{equation}
    \begin{aligned}[b]
        \sup_{\theta \in \Theta} \sum_{i=1}^I p_i^\epsilon R(d_i,\theta) &= \sup_{\theta \in \Theta} \sum_{i=1}^I \left(\frac{1}{T} \sum_{t=1}^T p_{i,t}\right) R(d_i,\theta)\\
        &= \sup_{\theta \in \Theta} \frac{1}{T}\sum_{t=1}^T \left(\sum_{i=1}^I p_{i,t} R(d_i,\theta)\right)\\
        &\leq \frac{1}{T}\sum_{t=1}^T \left(\sum_{i=1}^I p_{i,t}R(d_i,\theta_t)\right),
    \end{aligned}
    \label{eqn:step_2}
    \end{equation}
    where the inequality uses the fact that $\theta_t$ is nature's best response to $p_t$.

    Lemma~\ref{lemma:Lemma_convex} showed that $f(\cdot)$ is a continuous function on the compact set $\Delta^{I-1}$. Therefore, the minimax strategy of the decision problem exists, and we denote it as 
    \[p^* \in  \arg\min_{p \in \Delta^{I-1}} f(p). \]
    
     By Step 1, the right hand side of the Equation \eqref{eqn:step_2} is bounded by above by Equation \eqref{eqn:aux_step_1} for any $d_i$. It then follows that for any $\tilde{p}_i \in \Delta^{I-1}$, \eqref{eqn:step_2} is bounded above by
    \[\frac{1}{T}\sum_{t=1}^T \sum_{i=1}^{I} \tilde{p}_iR(d_i,\theta_t) + \frac{\epsilon}{2} + \frac{\ln(I)}{T} \left(\frac{M^2}{\epsilon}\right).\]
    In particular, we can use $p^*$ and further use the bound
    \[\sum_{i=1}^I p_i^* R(d_i,\theta_t) \leq \sup_{\theta \in \Theta} \sum_{i=1}^I p_i^* R(d_i,\theta) = \Bar{v}.\]
    Consequently, 
    \[\sup_{\theta \in \Theta} \sum_{i=1}^I p_i^\epsilon R(d_i,\theta) \leq \frac{1}{T}\sum_{t=1}^T \left(\sum_{i=1}^I p_{i,t}R(d_i,\theta_t)\right) \leq \Bar{v} + \frac{\epsilon}{2} + \frac{\ln(I)}{T} \left(\frac{M^2}{\epsilon}\right).\]
    This gives the desired result. 

    {\scshape STEP 3:} By taking the smallest integer T such that
    \[\frac{\ln(I)}{T} \left(\frac{M^2}{\epsilon}\right) \leq \frac{\epsilon}{2},\]
    or, equivalently, 
    \[T = \lceil2 M^2 \ln(I)/ \epsilon^2 \rceil.\]
    We then conclude that
    \[\Bar{v} \leq \sup_{\theta \in \Theta}\left(\sum_{i=1}^I p_i^\epsilon R(d_i,\theta)\right) \leq \frac{1}{T}\sum_{t=1}^T \left(\sum_{i=1}^I p_{i,t}R(d_i,\theta_t)\right) \leq \Bar{v} + \frac{\epsilon}{2} + \frac{\epsilon}{2} = \Bar{v}  + \epsilon,\]
    where $p_i^\epsilon \equiv \frac{1}{T}\sum_{t=1}^T p_{i,t}$ (as defined before). Since
    \[p^\epsilon \equiv (p_1^\epsilon,...,p_I^\epsilon) \in \Delta^{I-1},\]
    we conclude that $\Tilde{p}^\star$ is an $\epsilon$-minimax decision rule and that $\Bar{v}^{\epsilon} \equiv \frac{1}{T}\sum_{t=1}^T \left(\sum_{i=1}^I p_{i,t}R(d_i,\theta_t)\right)$ is an $\epsilon$-approximation to $\Bar{v}$.

\subsection{Theorem \ref{thm:MWU} with approximate evaluation of $f(p)$}\label{sec:approximate}

Suppose $\theta_t$, the  the exact solution of the problem in \eqref{eqn:f_convex}, is not available. Instead, we have access to $\theta_t^{\delta}$ such that \eqref{eq:theta.approximate} holds. The proof of Theorem \ref{thm:MWU} can be modified as follows.  Firstly, in Step 2, we say that
\[\sup_{\theta \in \Theta}\sum_{i=1}^I \Tilde{p}_i^\star R(d_i,\theta) \leq \frac{1}{T}\sum_{t=1}^T \left(\sum_{i=1}^I p_{i,t} R(d_i,\theta_t)\right).\]
We can change this to
\[\sup_{\theta \in \Theta}\sum_{i=1}^I \Tilde{p}_i^\star R(d_i,\theta) \leq \frac{1}{T}\sum_{t=1}^T \left(\sum_{i=1}^I p_{i,t}R(d_i,\theta_t^\delta) + \delta\right).\]
Then, by Step 1, this is bounded by above by
\[\frac{1}{T}\sum_{t=1}^T \sum_{i=1}^I p_{i,t} R(d_i,\theta_t^\delta) + \frac{\epsilon}{2} + \frac{\ln(I)}{T}\left(\frac{M^2}{\epsilon}\right) + \delta.\]
Then, subbing in to Step 2, the left hand side is bounded by above by
\[\Bar{v} + \frac{\epsilon}{2} + \frac{\ln(I)}{T}\left(\frac{M^2}{\epsilon}\right) + \delta.\]
Therefore, choosing $T$ as we have done gives an $\epsilon+\delta$ approximation. 

\subsection{Effect of Discretization}\label{app:decision-discretization}

\subsubsection{Discretizing the Space of Decision Rules}
Consider the treatment choice problem discussed in Section \ref{subsection:MMR}. Let $v^*\equiv\inf_{d\in\mathcal{D}^*}\sup_{\theta\in\Theta}R(d,\theta)$, where $\mathcal{D}^*$ is the set of all decision rules. Suppose there exists a  $d^*\in\mathcal{D}^*$ that achieves $v^*$ (i.e., $\sup_{\theta\in\Theta}R(d^*,\theta)=v^*$).  In this subsection, we study the approximation error of discretizing the space of decision rules for the minimax problem.

In the treatment choice problem discussed in Section \ref{subsection:MMR}, it is possible to show that there exists a minimax rule $d^*$ that can be written as a mixture of threshold rules with some mixture weights supported on $[-k,k]$. That is, we know that there exists a rule $d^*$ of the form:
\begin{equation}\label{eq:discretize.1}
d^*(\hat{\mu})\equiv\int d_c(\hat{\mu})dQ^*(c),    
\end{equation}
where  $d_c(x)=\mathbf{1}\{x\geq c\}$ and $Q^*(\cdot)$ is a univariate distribution supported on $[-k,k]$.
We approximate the value of $v^*$ by discretizing the decision space in the minimax problem.  In particular, fix a grid of $I$ equally spaced points in $[-k,k]$, which we write as $C_I=\{c_1,\ldots,c_I\}\subseteq[-k,k]$. Then, consider $\mathcal{D}=\{d_1,...d_I\}$, where $d_i(\hat{\mu})=\mathbf{1}\{\hat{\mu}\geq c_i\}$ and  $c_i=-k+\frac{i-1}{I-1}2k$. Let $\overline{v}_{\mathcal{D}}\equiv\inf_{d\in\Delta(\mathcal{D})}\sup_{\theta\in\Theta}R(d,\theta)$ and $L\equiv1/(\sigma\sqrt{2\pi})$. We have the following result.
\begin{lemma}\label{lem:discretization.1}
Suppose that \eqref{eq:discretize.1} holds. For any univariate  distribution $Q(\cdot)$ supported on $[-k,k]$, let $d_Q(\hat{\mu})\equiv\int d_c(\hat{\mu})dQ(c)$ be the mixture rule associated with $Q$. Suppose there exists some $\overline{\mu}<\infty$ such that, for every probability distribution $Q$ supported on $[-k,k]$,
\begin{equation}\label{eq:disc-worst-case-location}{
 \sup_{\theta\in\Theta}R(d_Q,\theta)
=\sup_{\theta\in\Theta_{\overline{\mu}}}R(d_Q,\theta),}
\end{equation}
where 
\begin{equation}\label{eq:Theta.restrict}
\Theta_{\overline{\mu}}=\{(\mu,\mu^*)\in\Theta:|\mu| \leq \overline{\mu}\}.
\end{equation}
Then, 
\[v^* \leq \overline{v}_{\mathcal{D}}\leq v^*+(\overline{\mu}+k)kL/(I-1).\]
That is, the minimax value of the discretized problem is at most $(\overline{\mu}+k)kL/(I-1)$ away from the minimax value of the original problem.
\end{lemma}

\begin{proof}

{\scshape{STEP 1:}} For any two thresholds $c,c^{\prime}\in[-k,k]$ and any $\theta\in\Theta$, note \begin{align}
 |R(d_c,\theta)-R(d_{c'},\theta)|
 &\notag\leq|\mu^*||c-c'|L,
\end{align}
where $L\equiv1/(\sigma\sqrt{2\pi})$, as the Gaussian c.d.f. $\Phi(\cdot)$ is Lipschitz continuous with constant $1/\sqrt{2\pi}$. 

{\scshape{STEP 2:}} Given grid $C_I$,  denote by $N(c)$ the nearest grid point to $c$ in $C_I$. If two such points exist, we break ties by choosing the smaller one. Then, define
$d_{Q^*_{I}}(\hat{\mu})\equiv\int d_{N(c)}(\hat{\mu})dQ^*(c)$. It follows that for any $\theta\in\Theta$, we have
\begin{align*}
\left| R(d_{Q_I^*}, \theta) - R(d^*, \theta) \right|
&= \left| \int \bigl( R(d_{N(c)}, \theta) - R(d_c, \theta) \bigr)
   \, dQ^*(c) \right| \\
&\le |\mu^*|L\int_{[-k,k]} \left| N(c)-c\right| dQ^*(c)\\
&\le |\mu^*|Lk/(I-1),
\end{align*}
where the last relation follows because, by construction,  $\left| N(c)-c\right|\leq k/(I-1)$ for any $c\in[-k,k]$. Conclude that 
\begin{equation}\label{pf.discretize.1}
R(d_{Q^*_I},\theta)
 \leq R(d^*,\theta)+ |\mu^*|Lk/(I-1),\quad\forall\theta\in\Theta.
\end{equation}

{\scshape{STEP 3:}} Note by construction,
$v^* \leq \overline{v}_{\mathcal{D}}\leq\sup_{\theta\in\Theta}R(d_{Q^*_I},\theta)$, since $d_{Q^*_{I}}$ can be viewed as randomizing among $\mathcal{D}$ with a particular weight in $\Delta(\mathcal{D})$. Invoking \eqref{eq:disc-worst-case-location} and \eqref{pf.discretize.1}, we have
\begin{align*}
v^* \leq\overline{v}_{\mathcal{D}}&\leq\sup_{\theta\in\Theta}R(d_{Q^*_I},\theta)=\sup_{\theta\in\Theta_{\overline{\mu}}}R(d_{Q^*_I},\theta)\leq \sup_{\theta\in\Theta_{\overline{\mu}}}R(d^*,\theta)+(\overline{\mu}+k)kL/(I-1) \\
&=\sup_{\theta\in\Theta}R(d^*,\theta)+(\overline{\mu}+k)kL/(I-1)=v^*+(\overline{\mu}+k)kL/(I-1).\end{align*}
\end{proof}
To use Lemma \ref{lem:discretization.1}, note that in the illustrative example of Section \ref{subsection:MMR}, we picked  $k=3$, $\sigma=1$, $\epsilon=0.15$ and $I=501$. It can be verified that \eqref{eq:disc-worst-case-location} holds with $\overline{\mu}=k$. Therefore, $(\overline{\mu}+k)kL/(I-1)=0.0144$. Thus, Lemma \ref{lem:discretization.1} implies that our algorithm provides a rule that is provably $0.1644$-minimax even if we target $v^*$.

\subsubsection{Discretizing Nature's Parameter Space}
\label{app:nature-discretization}
Next, we consider the effect of discretizing the parameter space for the treatment choice problem discussed in Section \ref{subsection:MMR}. The argument follows the same logic as in Lemma 2 of \cite{montielolea2026bounded}. Denote by $\pi(\cdot)$ a prior on $\Theta$. The collection of all (Borel) priors on $\Theta$ is referred to as $\Delta(\Theta)$. Given $\Delta(\mathcal{D})$, consider the maximin problem $\overline{v}_{\Theta}\equiv\sup_{\pi\in\Delta(\Theta)}\inf_{d\in\Delta(\mathcal{D})}\int R(d,\theta)d\pi(\theta)$.\footnote{Here, we consider a decision space $\Delta(\mathcal{D})$, which matches with the main results in the main text. The effect of discretizing the parameter space when the decision space is $\mathcal{D}^*$ can also be characterized analogously.}

Next, suppose we   approximate the value of $\overline{v}_{\Theta}$ by discretizing the parameter space in the maximin problem. We assume that there exists some $\overline{\mu}<\infty$ such that 
\begin{equation}\label{eq:discretize.3}
{
\sup_{\pi\in\Delta(\Theta)}\inf_{d\in\Delta(\mathcal{D})}\int R(d,\theta)d\pi(\theta)=\sup_{\pi\in\Delta(\Theta_{\overline{\mu}})}\inf_{d\in\Delta(\mathcal{D})}\int R(d,\theta)d\pi(\theta)},
\end{equation}
where $\Theta_{\overline{\mu}}$ is defined in \eqref{eq:Theta.restrict}. Then, we can discretize $\Theta_{\overline{\mu}}$ by  considering $\Theta_I=\{\theta_1,\ldots,\theta_I\}\subseteq\Theta_{\overline{\mu}}$. Let $\overline{v}_{\Theta_I}\equiv\sup_{\pi\in\Delta(\Theta_I)}\inf_{d\in\Delta(\mathcal{D})}\int R(d,\theta)d\pi(\theta)$ be the value of the discretized maximin problem. For each $\theta\in\Theta_{\overline{\mu}}$, let $\Delta(\theta)\equiv\min_{1\leq i\leq I}\Vert\theta_i-\theta\Vert_1$, where $\Vert\cdot\Vert_1$ refers to the vector $\ell_1$ norm. Write $\overline{\Delta}\equiv\sup_{\theta\in\Theta_{\overline{\mu}}}\Delta(\theta)$.  For example, for our application in the main text, with $k=3$ and $\sigma=1$, \eqref{eq:discretize.3} holds with $\overline{\mu}=k$. Let $I_1=501$ be the number of equally spaced grids in $\mu\in[-k,k]$  and $I_2=101$ be the number of equally spaced grids for each $\mu\in[\mu-k,\mu+k]$. It follows that $\overline{\Delta}=\max\{2k/(I_1-1),k/(I_2-1)\}=0.03$.  Define $L\equiv\max\{2,(\overline{\mu}+k)/(\sigma\sqrt{2\pi})\}$. The following result implies that our mirror ascent routine provably delivers an $(\epsilon+\overline{\Delta}L)$-least-favorable prior. 
\begin{lemma}\label{lem:discretize.2}
Suppose \eqref{eq:discretize.3} holds  and there exists some $d_{\Theta_I}^*\in\Delta(\mathcal{D})$ such that $\sup_{\theta\in\Theta_I}R(d_{\Theta_I}^*,\theta)=\overline{v}_{\Theta_I}$.  Then, $\overline{v}_{\Theta}\leq \overline{v}_{\Theta_I}+L\overline{\Delta}$.
\end{lemma}

\begin{proof}
 {\scshape{STEP 1:}} For each $d\in\Delta(\mathcal{D})$, note
$R(d,\theta)=\mu^*\mathbf1\{\mu^*\geq0\}-\mu^*\mathbb E_\mu[d(\hat{\mu})]$. 
Therefore, for any $\theta=(\mu,\mu^*)$ and $\theta^{\prime}=(\mu^{\prime},\mu^{*\prime})$, 
we have 
\begin{align*}
 |R(d,\theta)-R(d,\theta')|
 &\leq|\mu^*\mathbf1\{\mu^*\geq0\}-\mu^{*\prime}\mathbf1\{\mu^{*\prime}\geq0\}|\\
 &\quad+|\mu^*-\mu^{*\prime}|\,\mathbb E_\mu[d(\hat{\mu})]
       +|\mu^{*\prime}|\,|\mathbb E_\mu[d(\hat{\mu})]-\mathbb E_{\mu^\prime}[d(\hat{\mu})]|\\
 &\leq2|\mu^*-\mu^{*\prime}|
       +\frac{|\mu^{*\prime}|}{\sigma\sqrt{2\pi}}|\mu-\mu'|,
\end{align*}
where the last relation follows from $d(x)\in[0,1]$ for all $x\in\mathbb{R}$ and that 
\[
|\mathbb E_{\mu}[d(\hat{\mu})]-\mathbb E_{\mu^{\prime}}[d(\hat{\mu})]|\leq \frac{|\mu-\mu^{\prime}|}{\sigma\sqrt{2\pi}}.
\]
Note that for any $\theta^{\prime}\in\Theta_{\overline{\mu}}$, we have $|\mu^{*\prime}|\leq|\mu^{\prime}|+|\mu^{*\prime}-\mu^{\prime}|\leq \overline{\mu}+k$.  Conclude that  for any $\theta,\theta^{\prime}\in\Theta_{\overline{\mu}}$, we have
\begin{equation}\label{eq:disc-state-lipschitz}
 |R(d,\theta)-R(d,\theta')|
 \leq2|\mu^*-\mu^{*\prime}|+\frac{\overline{\mu}+k}{\sigma\sqrt{2\pi}}|\mu-\mu'|
 \leq L\|\theta-\theta'\|_1,
\end{equation}
where  $L=\max\{2,(\overline{\mu}+k)/(\sigma\sqrt{2\pi})\}$.

{\scshape{STEP 2:}} Given $\Theta_I$,  denote by $N(\theta)$ the nearest neighbor of $\theta$ in $\Theta_I$. In case of ties,  we break them by picking the element with the smallest index.  Applying \eqref{eq:disc-state-lipschitz}, we have 
\[
R(d_{\Theta_I}^*,\theta)\leq R(d_{\Theta_I}^*,N(\theta))+L\Delta(\theta),\quad \forall\theta\in\Theta_{\overline{\mu}}.
\]
But then, we have
\[
\sup_{\theta\in\Theta_{\overline{\mu}}}R(d_{\Theta_I}^*,\theta)\leq \sup_{\theta\in\Theta_{\overline{\mu}}}R(d_{\Theta_I}^*,N(\theta))+L\sup_{\theta\in\Theta_{\overline{\mu}}}\Delta(\theta)=\sup_{\theta\in\Theta_I} R(d_{\Theta_I}^*,\theta)+L\overline{\Delta}=\overline{v}_{\Theta_I}+L\overline{\Delta}.
\]
Invoking \eqref{eq:discretize.3}, we conclude that
\begin{align*}
\overline{v}_{\Theta}&=\sup_{\pi\in\Delta(\Theta)}\inf_{d\in\Delta(\mathcal{D})}\int R(d,\theta)d\pi(\theta)=\sup_{\pi\in\Delta(\Theta_{\overline{\mu}})}\inf_{d\in\Delta(\mathcal{D})}\int R(d,\theta)d\pi(\theta)\\
&\leq\inf_{d\in\Delta(\mathcal{D})}\sup_{\theta\in\Theta_{\overline{\mu}}}R(d,\theta)\leq\sup_{\theta\in\Theta_{\overline{\mu}}}R(d_{\Theta_I}^*,\theta)\leq \overline{v}_{\Theta_I}+L\overline{\Delta}.    
\end{align*}
\end{proof}

\section{Connection to S-Games}
\label{sec:S_games}

The minimax problem we study is closely related to what \cite{blackwell1954} call an $S$-game. Player I, the statistician in our case, has a finite number of pure strategies $d \in \mathcal{D} \equiv \{d_1,\ldots, d_I\}$. Player II, nature, may have infinitely many pure strategies. For each $\theta \in \Theta$, the strategy of nature can be represented as a vector in $\mathbb{R}^I$, $s(\theta) = (R(d_1, \theta), R(d_2, \theta), ..., R(d_I, \theta))$. Define 
\begin{equation}
\label{eqn:S_games_defS}
    S := \{ \left(R(d_1, \theta), R(d_2, \theta), ..., R(d_I, \theta)\right) \in \mathbb{R}^I| \theta \in \Theta\} 
\end{equation}
and $M(i, s) = s_i$, then $\Gamma_{P} = (\mathcal{D}, S, M)$ is a $S$-game with payoff matrix $M(i,s)$. The index $P$ stands for \emph{pure} as we are only considering pure strategies. The mixed extension of the $S$-game is equivalent to $\Gamma_{m} = (\Delta, R, M)$, where $\Delta$ is the set of discrete probability distribution over $\mathcal{D}$, and $R$ is the set of all countable convex linear combination of points in $S$. When $S$ is closed and convex, our minimax problem is exactly solving for the best mixed strategy for player I.
 
\cite[Theorem 2.4.2]{blackwell1954} states that i) Every $S$-game has a value, and the first player has a good (a minimax) strategy; and  ii) If $S$ is closed and convex, player II has a pure good strategy. 

Following these results, we have that i) there exists $\Bar{v}$, 
\begin{equation}
\label{eqn:minmax_thm_mixed}
    \inf_{p \in \Delta(D)} \sup_{q \in \mathcal{P}(\Theta)} \int_{\Theta} R(p, \theta) d q(\theta) = \sup_{q \in \mathcal{P}(\Theta)} \inf_{p \in \Delta(D)} \int_{\Theta} R(p, \theta) dq(\theta) = \Bar{v}, 
\end{equation}
where $\mathcal{P}(\Theta)$ is the set of all mixed strategies of the nature. And, there exists a \emph{minimax} decision rule $p^* \in \Delta(D)$ such that 
\begin{equation*}
    \sup_{\theta \in \Theta} R(p^*, \theta) = \Bar{v}.
\end{equation*}
ii) If the set S defined in \eqref{eqn:S_games_defS} is convex and closed, there exists $\theta^*$ that is \emph{maximin} strategy for nature, i.e.,  
\begin{equation*}
    \inf_{p \in \Delta(D)} R(p, \theta^*) = \Bar{v}.
\end{equation*}
Otherwise, the results in \cite{blackwell1954} show that the \emph{maximin} strategy for nature is supported on at most $I$ points. 

\subsection{Maximin Strategy (least-favorable distribution)}\label{sec:maximin.from.convex.routine}
Our proof for Theorem~\ref{thm:MWU} also gives a surprising result: when the game has a value, i.e., \eqref{eqn:minmax_thm_mixed} holds, we can derive approximate maxmin strategy for nature. 
\begin{definition}
A distribution $q^*_{\epsilon} \in \mathcal{P}(\Theta)$ is called an ``$\epsilon$-maximin" strategy for the game $(\Delta^{I-1}, \mathcal{P}(\Theta), R(\cdot, \cdot))$ with value $\bar{v}$ if 
    \begin{equation*}
        \inf_{p \in \Delta(D)} R(p, q^*_{\epsilon})
    \geq \sup_{q \in \mathcal{P}(\Theta)} \inf_{p \in \Delta(D)} R(p,q) - \epsilon = \Bar{v} - \epsilon,
    \end{equation*}
where 
\begin{equation*}
    R(p, q) := \int_{\Theta} R(p, \theta) d q(\theta).
\end{equation*}
\end{definition}
In our proof, we derived an intermediate result showing that: 
\begin{equation*}
    \frac{1}{T}\sum_{t=1}^T \left(\sum_{i=1}^I p_{i,t} R(d_i,\theta_t)\right) \leq \frac{1}{T}\sum_{t=1}^T \sum_{i=1}^{I} p_i R(d_i,\theta_t) + \frac{M^2 \eta}{2} + \frac{\ln(I)}{T \eta}, \forall p \in \Delta^{I-1},
\end{equation*}
see Equation \eqref{eqn:aux_step_1}.
Recall Assumption~\ref{asn:A2_oracle} states that for all $t$, $\sum_{i=1}^{I} p_{i,t} R(d_i, \theta_t) = \sup_{\theta \in \Theta} R(p_t, \theta)$, so 
\begin{equation*}
     \frac{1}{T}\sum_{t=1}^T \left(\sum_{i=1}^I p_{i,t} R(d_i,\theta_t)\right) 
     = \frac{1}{T}\sum_{t=1}^T \sup_{\theta \in \Theta} R(p_t, \theta)
     \geq \inf_{p \in \Delta^{I-1}} \sup_{\theta \in \Theta} R(p, \theta) = \Bar{v},
\end{equation*}
we get 
\begin{equation*}
    \Bar{v} \leq \frac{1}{T}\sum_{t=1}^T\sum_{i=1}^{I} p_i R(d_i,\theta_t) + \frac{M^2 \eta}{2} + \frac{\ln(I)}{T \eta} , \forall p \in \Delta^{I-1}
\end{equation*}
By taking $\eta = \epsilon/M^2, T = \lceil 2M^2 \ln(I) / \epsilon^2 \rceil$, we get 
\begin{equation*}
    \Bar{v} - \epsilon \leq \frac{1}{T}\sum_{t=1}^T\sum_{i=1}^{I} p_i R(d_i,\theta_t), \forall p \in \Delta^{I-1}. 
\end{equation*}
Now, if we choose $q^*_{\epsilon}$ to be a discrete distribution that 
\begin{equation*}
    q^*_{\epsilon}(\theta) = \frac{|\{t \in [T]: \theta_t = \theta\}|}{T}, 
\end{equation*}
Then, 
\begin{equation*}
    \inf_{p \in \Delta^{I-1}} R(p, q^*_{\epsilon}) = \inf_{p \in \Delta^{I-1}} \frac{1}{T}\sum_{t=1}^T\sum_{i=1}^{I} p_i R(d_i,\theta_t) \geq \Bar{v} - \epsilon = \sup_{q \in \mathcal{P}(\Theta)} \inf_{p \in \Delta^{I-1}} R(p,q) - \epsilon, 
\end{equation*}
which means $q^*_{\epsilon}$ is an ``$\epsilon$-maximin" strategy for nature.

\section{Additional Simulation and Empirical Results}

\subsection{Comparison with Fictitious Play from \cite{guggenberger2025numerical}}
\label{app:gp}
In this section, we compare the numerical performance of our hedge algorithm and the fictitious play algorithm proposed by \cite{guggenberger2025numerical} using two concrete examples, one from our Section \ref{subsection:MMR} on treatment choice with partial identification, and the other from the leading example in \cite{guggenberger2025numerical} on multi-arm treatment choice with RCT data.\footnote{Modifying traditional fictitious play as given in \cite{robinson1951iterative}, \cite{guggenberger2025numerical} use an alternative weighting scheme proposed by \cite{leslie2006generalised} that weights best responses at iteration $s$ according to $(5+s)^{-0.7}$, as opposed to $s^{-1}$ implicit in the traditional fictitious play of \cite{robinson1951iterative}.}  Overall, both algorithms performed very well in these simulations.  One of the main benefits of our algorithm is that it comes with concrete theoretical recommendations and guarantees (aside from the choice of $\epsilon$). One of the main benefits of fictitious play is that it tends to have lower memory usage requirements because the support points of least-favorable distributions and minimax strategies are added successively. The main caveat with fictitious play is---as discussed in \cite{guggenberger2025numerical}---that there are well-known problems in which fictitious play can perform quite poorly; see the reference to the work of \cite{daskalakis2014counterexample} in \cite{guggenberger2025numerical}. 

\subsubsection*{Treatment choice with partial identification}

We first apply the fictitious play algorithm to our treatment choice example in Section \ref{subsection:MMR}; see Table \ref{tab:comparison.1}. We considered two versions of our treatment choice example: i) one with a discretized space of decision rules, and ii) one with both our discretized decision space and a discretized parameter space - performing mirror ascent. We evaluated both algorithms by analyzing how quickly the candidate minimax rule and least-favorable distribution can be used to \emph{certify} a pair of $\epsilon$-minimax/maximin solutions for the \emph{discretized} problems. The certification compares the gap in the worst-case risk of the candidate minimax rule and the optimized average risk corresponding to the candidate least-favorable prior (see also Remark \ref{remark:early_stopping}). We say that the solution is certified if this gap is less than $\epsilon$. We set $\epsilon=.15$ as in our baseline example.

While the theoretical analysis for the Hedge algorithm makes a concrete suggestion for initializing the weights over decision rules or parameter values (uniform in both cases), there is no such analogous recommendation for fictitious play. For an immediate comparison, we first initialize fictitious play using uniform weights, referred to as FP-I1 in Table \ref{tab:comparison.1}. We also consider an initialization of fictitious play at a particular starting decision, where $c_i=0$ for mirror descent, $(\mu,\mu^*)=(0,0)$ for mirror ascent. Initialization at zero is referred to as FP-I2 in Table \ref{tab:comparison.1}. When the algorithms are applied to the problem with discretized decision rules, the certification with hedge and FP-I1 is almost immediate: 2 rounds with Hedge, and 14 with fictitious play. When the algorithms are applied to the problem with discretized parameter space the certification is slower (although much more so for the Hedge): 24 rounds for fictitious play and $4{,}604$ rounds with the Hedge. In this case, fictitious play still performs quite well and there is also an important memory advantage of fictitious play since the support points of the candidate minimax/maximin solutions are added iteratively. 

\begin{table}[htbp]
\centering
\begin{tabular}{@{}llrrrrr@{}}
\toprule
Menu & Method & $T_{\rm rec}$ & $U_T$ & $L_T$ & $U_T-L_T$
     & \shortstack{First round\\certified} \\
\midrule
Discretized      & Hedge & 9,974 & 1.500158 & 1.499790 & 0.000368 & 2 \\
decision space   & FP-I1 & 9,974 & 1.500978 & 1.499031 & 0.001947 & 14 \\
(mirror descent) & FP-I2 & 9,974 & 1.501014 & 1.499029 & 0.001984 & 24 \\
\addlinespace[6pt]
Discretized     & Hedge & 17,378& 1.500036& 1.451185& 0.048851& 4,604\\
parameter space & FP-I1    & 17,378& 1.500679& 1.499171& 0.001508& 24\\
(mirror ascent) & FP-I2 & 17,378& 1.500505& 1.499270& 0.001236& 32\\
\bottomrule
\end{tabular}
\caption{Hedge versus fictitious play in treatment choice with partial identification of Section \ref{subsection:MMR}}
\label{tab:comparison.1}
{\raggedright\footnotesize
\textit{Notes}: FP-I1 refers to the fictitious play algorithm recommended
in \cite{guggenberger2025numerical} but with uniform weights for
initialization, whereas FP-I1 refers to fictitious play but with a
particular starting decision for initialization.
$T_{\text{rec}}$ is the number of rounds recommended by our hedge
algorithm to guarantee $\epsilon$-minimax/maximin optimality, which we
also apply to both versions of fictitious play.
$U_t$ stands for the worst-case risk of the candidate minimax rule at
round $t$; $L_t$ is the optimized average risk corresponding to the
candidate least-favorable prior at round $t$.\par}
\end{table} 
\subsubsection*{Multi-arm treatment choice based on experimental data}

Now, we apply our Hedge algorithm using mirror \emph{ascent} to the leading example in \cite{guggenberger2025numerical} that studies  multi-arm treatment choice problems based on experimental data. Here, one concern is about the dimensionality of the grid for the parameter space. To understand the extent to which this would be an issue, we first considered a model with two treatment arms and used the grid recommended by \cite{guggenberger2025numerical}, which has $I=1{,}001^2$ points. In this problem, the risk function takes values in the $[0,1]$ interval, so our Assumption \ref{asn:A1_Risk} is verified with $M=1$. For our Assumption \ref{asn:A2_oracle}, the best-response oracle is the same as the oracle used by \cite{guggenberger2025numerical} (see their Comment 1 to their Algorithm 1).  We take $\epsilon=.1$. This gives a recommended number of rounds of $\lceil2 \ln(I)/\epsilon^2 \rceil =2,764$. As a leading case for comparison, we consider the case in which the number of individuals in each treatment arm is $(N_1,N_2)=(10,20)$. This is one of the many instances analyzed by \cite{guggenberger2025numerical}.\footnote{Since there is no general recommendation for the number of iterations for fictitious play, we also use $2,764$ iterations for fictitious play.} See Panel A in Table \ref{tab:comparison.2}. Running the Hedge in this problem takes about 90 seconds. Our implementation of their recommended fictitious play algorithm is also quite fast and it takes about 20 seconds.\footnote{We follow their suggestion and initialize the algorithm  with a simple``empirical success rule''.} Hedge produces a certifiable $\epsilon$-least-favorable distribution on the parameter grid, with optimized average risk 0.023505. Randomizing uniformly over the best-response rules $d_t$ gives an $\epsilon$-minimax rule with worst-case risk 0.032737. These provide lower and upper bounds on the grid minimax value. Fictitious play gives tighter bounds of 0.032726 and 0.032728, respectively.\footnote{One possible explanation for the better performance of fictitious play is that a least-favorable prior in this problem may be sparse. \label{footnote:sparse}}. This shows that our algorithm delivers what it promises. As discussed in the paper, $\epsilon=.1$ should be viewed as an upper bound on the gap between the candidate maximin and minimax rules (in this example, the gap is almost a tenth of the value of $\epsilon$ we picked).

\begin{table}[ht]
\centering
 \begin{tabular}{@{}lrrrrrr@{}}\toprule
Method & Rounds & Minutes & \shortstack{Peak memory\\(GB)} & $U_t$ & $L_t$ & $U_t-L_t$ \\
\midrule
\multicolumn{7}{@{}l}{\textit{Panel A: Two arms},\quad  $I=1{,}001^2$, \quad $(N_1,N_2)=(10,20)$}\\[1pt]
FP & 2,764 & 0.34 & 0.71 & 0.032728 & 0.032726 & 0.000002 \\
Hedge & 2,764 & 1.58 & 0.72 & 0.032737 & 0.023505 & 0.009232 \\
Streamed Hedge & 2,764 & 1.23 & 0.69 & 0.032737 & 0.023505 & 0.009232 \\
\addlinespace[6pt]
\multicolumn{7}{@{}l}{\textit{Panel B: Three arms}, \quad $I=101^{3}$, \quad $(N_1,N_2,N_3)=(5,5,5)$}\\[1pt]
FP & 2,770 & 3.24 & 0.68 & 0.084156 & 0.083411 & 0.000745 \\
Hedge & 2,770 & 9.63 & 0.70 & 0.083897 & 0.068633 & 0.015264 \\
Streamed Hedge & 2,770 & 7.32 & 0.69 & 0.083897 & 0.068633 & 0.015264 \\
\addlinespace[6pt]
\multicolumn{7}{@{}l}{\textit{Panel C: Three arms}, \quad $I=1{,}001^{3}$, \quad $(N_1,N_2,N_3)=(5,5,5)$}\\[1pt]
FP & 20 & 6.91 & 0.72 & 0.130536 & 0.044113 & 0.086422 \\
Hedge & 20 & 27.99 & 16.74 & 0.084472 & 0.042009 & 0.042463 \\
Streamed Hedge & 20 & 15.98 & 0.71 & 0.084472 & 0.042009 & 0.042463 \\
\bottomrule\end{tabular}
    \caption{Hedge versus fictitious play in multi-arm treatment choice with experimental data}
    \label{tab:comparison.2}
    {\raggedright\footnotesize
\textit{Notes}: FP refers to the fictitious play algorithm recommended
in \cite{guggenberger2025numerical}, initialized at the empirical success rule.
$U_t$ stands for the worst-case risk of the candidate minimax rule at
round $t$, evaluated according to the specified number of grid points in the parameter space; $L_t$ is the optimized average risk corresponding to the
candidate least-favorable prior at round $t$.\par}
\end{table}

As noted by \cite{guggenberger2025numerical}, the main challenge with using the Hedge algorithm for more than two treatment arms, is that the Hedge updates need to keep track of a vector of dimension $I$. With three treatment arms, a grid of $1{,}001$ points in each dimension requires storing a vector of dimension $(1{,}001)^3$. We consider the following two alternative approaches to deal with this problem. 

First, in light of footnote \ref{footnote:sparse}, the potential sparsity of the least-favorable distribution suggests that we could, in principle, reduce the number of grid points in each dimension and use a grid with only $I=101^3$. We run mirror ascent using that grid, generate an approximate least-favorable distribution and a candidate approximate minimax rule.\footnote{One may evaluate the worst-case risk of a candidate rule over the larger grid with $I=1{,}001^3$ points.} In this example the Hedge generates a gap between the minimax and maximin values below $\epsilon$. See Panel B in Table \ref{tab:comparison.2}.

Second, to reduce the memory usage requirement,  we suggest a simple numerical approach to \emph{stream} the Hedge by keeping track of the history of Bayes decision rules; see 
Appendix \ref{app:streamed} for details of this streamed version of the Hedge algorithm for mirror descent.  Streamed Hedge and the standard Hedge give exactly the same answers, but their implementation and memory requirements are different. In fact, there are large values of $I$ for which we cannot complete a single round of the standard Hedge algorithm, but we can still complete a few rounds of streamed Hedge. As we see in Panel C of Table \ref{tab:comparison.2}, when applied to the problem of $I=1{,}001^3$, the streamed Hedge has similar memory requirements as fictitious play, about .72GB of RAM, compared to 16.74GB for ordinary Hedge. Fictitious play is still faster (about 7 minutes for 20 iterations, compared to 16 minutes for the streamed Hedge). While streaming the Hedge reduces maximum memory required, the run-time of both hedge and fictitious play is compromised by the large dimensionality of $I$.

\subsection{Streamed Hedge Algorithm}
\label{app:streamed}

One important limitation of the Hedge algorithm is that it needs to keep track of a vector of dimension $I$ (the candidate mixed strategies) to implement the updates. This storage is computationally costly. However, there is a way for us to modify the implementation of our algorithm to reduce this storage requirement. Algorithm \ref{alg:streamed-hedge-general} implements exactly this method, which we refer to as \textit{streamed} Hedge, and this algorithm gives numerically identical answers compared to the standard Hedge. In the case of mirror descent, the reduction in memory comes from redefining the input of the best response oracle. Instead of storing a vector of dimension $I$, we store the history of best responses from nature $\mathcal H_t=\{\theta_1,\ldots,\theta_{t}\}$ and give these as an input to the best-response oracle. Once we call the best-response oracle, we use algebra to analytically recover each of the candidate mixed strategies in a loop and use them to compute nature's best response.  More precisely, in Algorithm \ref{alg:mirror_descent} step 8, the probability distribution for action $i$, $p_{it}$ would be updated per $p_{it} \propto p_{i,t-1} \exp(-\eta R(d_i,\theta_{t-1}))$. However, to avoid having to store the entire vector $\{p_{i,t}:i\in \{1,\ldots,I\}\}$ which grows at a rate $O(I)$, the value of $p_{i,t}$ can be reconstructed on-demand using only $\mathcal H_t$ according to $p_{i,t}\propto \exp\left\{-\eta\sum_{s=1}^{t-1} R(d_i,\theta_s) \right\}$. This observation motivates the following algorithm.

\begin{algorithm}[H]
\caption{Streamed Hedge for the $\epsilon$-minimax problem
         with general $\Theta$}
\label{alg:streamed-hedge-general}

\linespread{1}\selectfont
\setlength{\abovedisplayskip}{4pt}
\setlength{\belowdisplayskip}{4pt}
\setlength{\abovedisplayshortskip}{4pt}
\setlength{\belowdisplayshortskip}{4pt}

\begin{algorithmic}[1]
\setlength{\itemsep}{2pt}
\setlength{\parsep}{0pt}

\Require Decision rules $\mathcal D=\{d_1,\ldots,d_I\}$;
parameter space $\Theta$; risk routine $R(d_i,\theta)$;
bound $M$; target approximation error $\epsilon\leq M$.

\State Set $\eta\gets\epsilon/M^2$ and
$T\gets\left\lceil 2M^2\ln(I)/\epsilon^2\right\rceil$.

\State Initialize the history $\mathcal H_0\gets\varnothing$.

\For{$t=1,\ldots,T$}
    \State Define the streamed objective\footnote{As written, our critical step 4 makes the algorithm perform
    $O(T^2)$ risk evaluations. This is a severe drawback, but with multiple easy solutions.
    If $\Theta$ is finite, we can instead count how often each
    parameter value has occurred and compute cumulative risk
    by multiplying each value's risk by its count and summing.
    In our multiple-treatment-arm example, which uses mirror
    ascent, we store the running sum of past decision rules
    in a table over possible observations; this table is of size $S=\Pi_{s=1}^A(n_s+1)$, the same dimensionality as the decision rule $d$ as described in \cite{guggenberger2025numerical}. Linearity of regret
    then lets us compute cumulative regret directly from this
    table. For fixed $|\Theta|$ or fixed table size, these
    shortcuts make the history-related work linear in $T$.
    They are not shown in this pseudocode; grid scans and
    best-response calculations retain their own costs.}
    \[
        F_t(\theta;\mathcal H_{t-1})
        \equiv
        \sum_{i=1}^I
        \exp\left\{-\eta\sum_{s=1}^{t-1}
        R(d_i,\theta_s)\right\}R(d_i,\theta).
    \]

    \State Compute
    \[
        \theta_t\in\arg\sup_{\theta\in\Theta}
        F_t(\theta;\mathcal H_{t-1}).
    \]

    \State Append $\theta_t$ to the history:
    $\mathcal H_t\gets(\mathcal H_{t-1},\theta_t)$.
\EndFor

\State Compute and retain the log-normalizers
$\lambda_1,\ldots,\lambda_T$, where
\[
    \lambda_t\equiv
    \log\left(
        \sum_{i=1}^I
        \exp\left\{-\eta\sum_{s=1}^{t-1}
        R(d_i,\theta_s)\right\}
    \right).
\]
Alternatively, compute and store each log-normalizer before step 5

\State \Return $\mathcal H_T=(\theta_1,\ldots,\theta_T)$
and $(\lambda_1,\ldots,\lambda_T)$. These outputs determine an an $\epsilon$-decision rule and an $\epsilon$-least favorable prior:

\[
\bar{w}_{T,i}
=
\frac{1}{T}\sum_{t=1}^{T}
\exp\left\{
-\eta\sum_{s=1}^{t-1}R(d_i,\theta_s)-\lambda_t
\right\},
\qquad i=1,\ldots,I.
\]

\[
\bar{d}_T(x)
=
\sum_{i=1}^{I}\bar{w}_{T,i}\,d_i(x).
\]

\[
\bar{\pi}_T
=
\frac{1}{T}\sum_{t=1}^{T}\delta_{\theta_t}.
\]

\end{algorithmic}
\end{algorithm}

The oracle requirement is unchanged. Let $S_\Theta+1$ be the
storage for one best response and one log-normalizer per round, while $S_{\mathrm{BR}}$ denotes the oracle's
working memory. Retaining the response history and normalizers
gives peak working memory
\[
 O(T (S_\Theta+1)+S_{\mathrm{BR}})
 =O\!\left(\frac{M^2\ln I}{\epsilon^2}+S_{\mathrm{BR}}\right)
 \quad\text{when }S_\Theta=O(1).
\]

\subsection{Sensitivity to the Lipschitz constant in Section \ref{sec:application} }

\label{app:lipshitz_sensitivity}

We compare $C=1.1834$ and $C=4.9425$ for one and five policy-relevant
sites, holding fixed the three experimental sites and the common
standard error $\sigma=4.5$. Policy sites are ordered by distance
from Dhaka--Pabna. The one-site case uses Dhaka--Faridpur; the five-site
case also includes Chittagong--Jhalokati, Dhaka--Narsingdi,
Dhaka--Barguna and Dhaka--Kishoregonj.
Table~\ref{tab:C-sensitivity} reports the averaged mixtures,
completed iteration counts and runtimes using $\epsilon=.1$.
\begin{table}[htbp]\centering\small
\captionsetup{font={small},labelfont={bf}}
\caption{Sensitivity of the computed site-selection mixtures.}
\label{tab:C-sensitivity}
\begin{tabular}{rrrrr}
\hline
\multicolumn{5}{l}{A. Averaged selection probabilities}\\
Policy sites & $C$ & Gaibandha & Noakhali & Pabna\\
\hline
1 & 1.1834 & 0.026905 & 0.030356 & 0.942738\\
1 & 4.9425 & 0.004345 & 0.004580 & 0.991075\\
5 & 1.1834 & 0.231842 & 0.039501 & 0.728657\\
5 & 4.9425 & 0.204646 & 0.006918 & 0.788436\\
\hline
\multicolumn{5}{l}{B. Computational details}\\
Policy sites & $C$ & $M$ & Iterations & Seconds\\
\hline
1 & 1.1834 & 2.240163 & 1,103 & 7.14\\
1 & 4.9425 & 11.731436 & 30,240 & 195.01\\
5 & 1.1834 & 2.148604 & 1,015 & 10.73\\
5 & 4.9425 & 11.018249 & 26,675 & 324.53\\
\hline
\end{tabular}
\par\smallskip
\parbox{.96\linewidth}{\footnotesize All three experimental corridors have Dhaka as their
destination. Probabilities are rounded independently.
All four cases use $T=\lceil2M^2\log(3)/.1^2\rceil$ and $\eta=.1/M^2$,
with $M$ computed by maximizing
the objective in Equation \eqref{eq:site-selection-minimax-problem}
over both $p\in\Delta^2$ and
$\tau\in\text{Lip}_{C}(\mathbb{R}^{13})$. Seconds measure the Hedge loop, including
the numerical best-response calculations, and exclude calculating $M$.}
\end{table}

Dhaka--Pabna receives the largest probability in all four examined cases and results from using $C=4.9425$ do not differ materially from $C=1.1834$. 
 \end{document}